\documentclass[reqno]{amsart}

\usepackage{amsmath,amssymb,amsfonts,amsthm}

\usepackage{mathtools}
\usepackage{xcolor}
\usepackage{mathpazo}
\usepackage{enumitem} 
\usepackage{graphicx}
\usepackage{subcaption}

\usepackage[linkcolor=blue,colorlinks=true, citecolor=blue, urlcolor=blue]{hyperref}

\usepackage[a4paper, total={6in, 8in}]{geometry}

\newenvironment{emphasisquote}
{\begin{quote}}
{  \emph{(emphasis added)}\end{quote}}

\title[Misuse of mathematics in the Hong-Page framework]{Fatal errors and misuse of mathematics in the Hong-Page Theorem and Landemore's epistemic argument}
\author{Álvaro Romaniega}
\email{alvaroromaniega@gmail.com}
\thanks{I would like to thank Jason Brennan, Lu Hong, Hélène Landemore and Scott E. Page for their valuable comments after the manuscript was completed. I would also like to thank \href{https://mvazcar.com/}{Miguel Vázquez-Carrero} for his thorough  revision of the manuscript.}

\newtheorem{theorem}{Theorem}[section]

\newtheorem{corollary}[theorem]{Corollary}
\newtheorem{proposition}[theorem]{Proposition}

\theoremstyle{definition}
\newtheorem{definition}[theorem]{Definition}

\newtheorem{assumption}{Assumption}
\newtheorem{error}{Error}

\newtheorem{remarkx}[theorem]{Remark}
\newenvironment{remark}
{\pushQED{\qed}\begin{remarkx}}
	{\popQED\end{remarkx}}

\begin{document}

	\begin{abstract}
In the pursuit of understanding collective intelligence, the Hong-Page Theorems have been presented as cornerstones of the interplay between diversity and ability. However, upon rigorous examination, there seem to be inherent problems and misinterpretations within these theorems. Hélène Landemore's application of these results in her epistemic argument and her political proposal showcases a invalid use of mathematical principles. This paper critically dissects the Hong-Page Theorems, revealing significant inconsistencies and oversights, and underscores the indispensable role of ``ability" in group problem-solving contexts. This paper aims not to undermine the importance of diversity, but rather to highlight the dangers of misusing mathematical principles and the necessity for a more careful analysis of mathematical results when applying them to social sciences.

	\end{abstract}

	\maketitle
	\tableofcontents\label{sec:toc}

In the burgeoning field of collective intelligence theory, one theorem has emerged as particularly influential: the Hong-Page Theorem. This theorem is claimed to show that in certain conditions, a diverse group of problem solvers has the potential to outperform a homogeneous group of high-ability problem solvers. The theorem has been instrumental in shaping discourses around the importance of diversity in decision-making and problem-solving contexts. One notable application of the Hong-Page Theorem is found in the work of Hélène Landemore, who uses it as the cornerstone of her political proposal for an ``Open Democracy." Landemore's thesis argues that greater cognitive diversity in a collective decision-making process not only enhances its epistemic properties, achieving the epistemic superiority of democracy, but also serves as a normative benchmark.

However, the theorem has not gone unchallenged. Mathematician Abigail Thompson has been among the most vocal critics, casting doubt on the use of the Hong-Page Theorem in this way and raising concerns about the validity of the conclusions drawn from it. In this paper, we aim to shed light on the Hong-Page Theorem, presenting our view that its findings, while intriguing, may not hold the significant social implications often attributed to them and its central insight may be less informative than widely believed.  By doing this, we aim to reveal the misconceptions propagated through their widespread acceptance, especially in their application to real-world socio-political phenomena, such as in Hélène Landemore's ``Open Democracy" proposition and her epistemic argument for democracy.

Thompson's article ought to have highlighted the problems with the Hong-Page theorem. However, this has not been the case given its widespread acceptance, several rejoinders (particularly focusing on the simulation part) and papers building on Hong-Page's work. For instance, after the publication of \cite{Tho14}, Landemore continues to base her argument on the theorem, \cite{BL21}, for unreplied rejoinders to Thompson's critique see \cite{Kue17} or \cite{Sin19} and for an example of a paper building on Hong-Page's work see \cite{GSea19}. Here, we aim to respond to these criticisms and the usage of the theorem by delving deeper into the theoretical aspects and using this to analyze particular misapplications. Our new foundational critiques and examples of misapplications seriously compromise the Hong-Page theorem and these works around it.
Therefore, in this article, we introduce new arguments that underscore the theorem's simple insight and its misapplications. While Thompson's critique focused on illustrating the theorem's simplicity through specific examples, offering amendments to the theorem's statement, and assessing the role of simulations, this article examines several critical dimensions more deeply:

\begin{itemize}
	\item Demonstrating the theorem's simplicity by introducing and proving straightforward and non-probabilistic versions of it. These versions encapsulate all arguments related to diversity or ability, devoid of hypotheses that are not tied to diversity or ability, e.g., introducing clones and a random selection of them.
	\item 
	Highlighting the theorem's simplicity by offering a concise proof, based on standard mathematical results, thus calling attention to how the deployment of complex mathematics can sometimes obscure the straightforward truths.
	\item Unveiling further, arguably unrealistic, consequences of the hypotheses of the theorem that cast doubt on their plausibility. For example, the random group always arriving at the correct solution unanimously.
	\item Explicitly defining all the hypotheses and auxiliary assumptions which helps to illuminate how they seem to naturally lead to a predetermined conclusion, even if this was not the original intention of the theorem itself. This approach will be partly illustrated through the use of "counterexamples" to the theorem without these assumptions, i.e., diversity does not trump ability when some technically tailored hypotheses are removed.
	\item Enhancing the discussion further, we adopt a framework akin to the original, yet grounded in more empirically valid and theoretically sound assumptions (for example, excluding the use of clones). This refined approach maintains a formal resemblance to the original theorem, to the extent that it could be viewed as a sophisticated enhancement. However, it arrives at fundamentally different conclusions, demonstrating that when realistic hypotheses are used, ability significantly outweighs diversity.
	\item Investigating other theorems by Hong and Page, particularly the ``Diversity Prediction Theorem," to highlight similar issues. This section aims to clarify the limitations of applying such propositions in real-world contexts. For the sake of clarity we have moved it to an appendix.
	\item Identifying the misapplications of the theorem by Hong and Page in various papers, books, and conferences. See also below for a discussion of Page's misstatement of the theorem, in the context of Landemore's work.
	\item Similarly as in the previous bullet, examining the (mis)use of the theorem in the political proposals of authors like Hélène Landemore. More specifically, we analyze:
	\begin{itemize}
		\item The misstatements of the thesis of the theorem, in the sense that  it does not support what is claimed.
		\item The misstatement of the hypotheses of the theorem, in the sense that the necessary hypotheses are significantly more than the ones stated. This is also present in Page's statement of the theorem. Our analysis also serves as a rebuttal to Page's response to Thompson's critique.
		\item Assessing the realism of the theorem's hypotheses within political applications through a Moorean style argument.
		\item Assessing the realism of the theorem's hypotheses within political applications  highlighting potential contradictions.
		\item Evaluating the substantive content of Landemore's derived results, such as the ``Numbers Trump Ability Theorem."
		\item Emphasizing the importance of ability in the previous theorem. Properly applied, it would support a version of epistocracy.
	\end{itemize}
\end{itemize}

The paper is organized as follows. We begin with a thorough dissection of the Hong-Page ``Diversity Trumps Ability" Theorem in Section \ref{sec:2}. We delve into the definitions that underpin these theorems and carefully examine the assumptions relating to the problems and problem solvers. We then derive and discuss a series of trivial corollaries from these assumptions, concluding with a concise analysis of other related results and a simplified proof of the theorem.

In Section \ref{sec:3}, the paper takes a critical turn, presenting a series of counterexamples when some arbitrary hypothesis are removed that challenge the robustness of the Hong-Page Theorem. We start by challenging the necessity of the injective function and the ``unique best agent'' assumption, eventually moving towards a discussion on the performance and selection of clones. We conclude that what the theorem requires is not ``diversity", but the existence of a more ``able" problem solver who can improve upon areas where others fall short.

The critique deepens in Section \ref{sec:ability trumps div} with the proposition of a new Hong-Page style theorem: ``Ability trumps diversity''. This section revises the original assumptions to allow a fair comparison between ability and diversity, ultimately culminating in the presentation of this theorem.

Moving into Section \ref{sec:6}, we delve into the misuse  of mathematics in the original Hong-Page theorems. Here, ``misuse" refers to the objective fact that mathematics is employed as a step in an argument in an incorrect way (that is, logically invalid or unsound argument), without delving into the reasons behind this misapplication. We explore how the application of ``complex" mathematics might have been simplified making its relatively simple insights more accessible. Additionally, we investigate instances of misapplications of the theorem. We further discuss the issue of using the prestige of mathematics to lend credence to their interpretation.

In the penultimate section, Section \ref{sec:7}, we turn our critique towards Hélène Landemore's political proposal. We delve into the conceptual stretch of the theorem when applied in political philosophy. Specifically, we suggest that the mathematical requisites of both the original theorem and her version are not met, which compromises the conclusions drawn by Landemore. In this context, we examine the application of mathematics and identify misinterpretations of the mathematical theorems. Further, we scrutinize the application of hypotheses in the ``Diversity Trumps Ability Theorem'' and discuss foundational issues of her ``Numbers Trump Ability Theorem''.

Finally, Section \ref{sec:8} wraps up the paper, offering a robust summary of the findings and their implications on the ongoing discourse surrounding collective intelligence and the role of diversity and ability within it. 

The initial sections of this text employs minimal mathematical language to facilitate understanding of the theorem, its proof, and corollaries. Wherever mathematical language is used, we have endeavored to provide a non-technical explanation, typically prefaced by the phrase, ``In other words". However, Sections \ref{sec:ability trumps div}, \ref{sec:6}, \ref{sec:7}, and \ref{sec:8} primarily reference earlier sections (or the appendices) and the results derived therein, and are essentially devoid of mathematics. To improve accessibility, the more technical parts are moved to the appendices. We have decided to keep those technical parts for those interested in mathematical or other details, while enabling a lighter reading that excludes these appendices. The mathematical parts presuppose a degree of mathematical proficiency (although the mathematics used are not overly complex). In particular, Appendix \ref{sec:5} shifts focus onto the ``Diversity Prediction Theorem'' and the ``Crowds Beat Averages Law'', discussing their results, pointing out the asymmetric role of ``ability'' and ``diversity'' and highlighting a basic mathematical error made by Page in advocating for diversity. This is a related theorem often cited with the ``Diversity Trumps Ability Theorem", but we have moved it to an appendix to avoid introducing unnecessary mathematical details into the text, although some misconceptions around these theorems will be discussed in the main text.

We hope that our critique stimulates introspection within the collective decision theory community and encourages a more discerning and judicious use of mathematical theorems. Our discussion aims to encourage a mindful approach to mathematical application in the collective decision theory realm, underscoring the necessity for clarity and detailed analysis. This ensures mathematics serves to illuminate rather than obscure.

A final caveat, this critique should not be taken as a dismissal of the importance of diversity (which we consider to be one important epistemic factor among others) in decision-making, but rather as a call to address the misuse of mathematics in these contexts. It urges us to consider the rigorous and nuanced approach required when applying mathematical theories to sociopolitical constructs. As such, this paper makes a contribution to the ongoing discourse on collective intelligence, fostering a deeper understanding of the mathematical theorems used.

\section{The ``Diversity Trumps Ability" Theorem}\label{sec:2}
This section introduces the key definitions and assumptions utilized in our analysis of the Hong-Page problem-solving framework.
\subsection{Definitions}
Here, we define what constitutes a problem solver and the nature of problems they tackle.
\begin{definition}[Problem solver] 
	Let $X$ be a set of possible states of the world. For simplicity, $X$ is assumed to be a finite set. A \textit{problem solver} is a function $\phi:X\to X$. We will consider several problem solvers throughout this analysis. The set of problem solvers is denoted by $\Phi$.
\end{definition} 
\begin{definition}[Problem]
	For each problem solver $\phi$, define the value function as $V_\phi: X \to [0,1]$ as a function mapping states of the world to real numbers between 0 and 1. The \textit{problem} is finding the maximum of $V_\phi$. Given a probability measure on $X$ with full support, $\nu$, the expected value of the performance of each agent is given by
	\begin{equation}\label{eq:expected value}
		\mathbb{E}_\nu(V_\phi\circ\phi)=\sum_{x\in X}V_\phi\left(\phi(x)\right)\nu(x)\,.
	\end{equation}
\end{definition}
Let's discuss the intuition behind the problem-solving process. Given an initial state\footnote{This is the equivalent of the  guess value or starting value in an optimization algorithm.} $x$ from the set $X$, a problem solver aims to find a solution to the problem by mapping $x$ to $\phi(x)$. In other words, the problem solver transforms the input $x$ into a potential solution, hoping that this transformation will maximize the value of the function $V_\phi$.

\subsection{Problem assumptions}
Here, we state the assumptions implicit in the Hong-Page framework. In particular, we are going to impose some conditions on the set $\Phi$ and on the problem following Hong-Page's assumptions.

\begin{assumption}[Unique problem]\label{as:unique prob} All problem solvers share the same value function, i.e.,
$\forall \phi,\phi'\in \Phi,~~ V_\phi=V_{\phi'}=:V.$ 
\end{assumption}
 In other words, any two problem solvers evaluate the problem space in the same way. This assumption simplifies the analysis by ensuring that all problem solvers have a consistent evaluation criterion for the problem. See Remark \ref{rem:Rawls} for why this assumption is not shared in standard political philosophy.

	\begin{assumption}[Unique solution]
		There exists a unique\footnote{Hereafter, the subscript in the existential quantifier denotes the cardinality of the subset of elements satisfying a given property. For instance, $\exists_{=1}$, sometimes denoted as $\exists!$, means that there is only one element satisfying the required condition.} state of $X$ that maximizes the value function $V$, i.e., 
		$\exists_{=1} x^*~/~ V(x^*)=1\,.$
	\end{assumption}
In other words, there is exactly one optimal solution to the problem that maximizes the value of the function $V$. This assumption allows us to focus on finding the unique solution.
	\begin{assumption}[Strictly increasing problem]\label{as:injective}
	$V$ is injective, i.e., if $V(x)=V(x')$, then $x=x'.$		
    That is, we can order $X$ as $\{x_1,\ldots,x_{|X|}\}$ such that
    $$
    V(x_1)<\ldots<V(x_{|X|})\,.
    $$
    \end{assumption}
In other words, this assumption implies that the problem has a well-defined ordering of potential solutions. The original article did not state explicitly that the value function V is one-to-one. However, this assumption is necessary for the theorem to hold, as Thompson pointed out, \cite{Tho14}, see Section \ref{sec:V inj} for more details.
	\subsection{Problem solver assumptions}
	\begin{assumption}[Everywhere ability in problem solvers]\label{as:ability} For all problems solvers $\phi\in\Phi$ it holds that
		$V(\phi(x))\ge V(x)$ for every $x\in X$. In particular, $\phi(x^*)=x^*\,.$	
	\end{assumption}
	This assumption states that all problem solvers are able to improve (with non-strict inequality) the value of any state. Furthermore, by Assumption \ref{as:injective}, if $\phi(x)\neq x$, $V(\phi(x))>V(x)\,.$ In other words, if a problem solver is applied to a state, the value of the state will never decrease. Furthermore, if the agents transforms the state, the value will strictly increase.
	\begin{assumption}[No improvement, idempotence]\label{as:idempot}
	$\forall \phi\in\Phi,~~ \phi\circ\phi=\phi\,.$		
	\end{assumption}
	This assumption states that problem solvers are idempotent. In other words, applying a problem solver to a state twice will have the same effect as applying it once, without affecting the value of the candidate/potential solution.

	\begin{assumption}["Difficulty", imperfect problem solvers]\label{as:difficult} For every problem solver $\phi$ there exists a state $x$ such the optimal solution, $x^*$, is not returned, i.e.,
		$\forall \phi\in\Phi~\exists x\in X ~/~ \phi(x)\neq x^*\,.$		
	\end{assumption}
In other words, by hypothesis, for every agent there are instances where they fail to find the optimal solution.
	\begin{assumption}[``Diversity'', sufficient unstuck problem solvers]\label{as:diversity}
	For every state $x\neq x^*$ there exists a problem solver $\phi$ such that it returns a different state of the world, i.e.,
	$\forall x\in X\backslash\{x^*\}~\exists~ \phi\in\Phi ~/~ \phi(x)\neq x\,.$		
	\end{assumption}
In other words, this assumption ensures\footnote{As Hong and Page put it.} a ``diversity'' of problem solvers in $\Phi$, with at least one problem solver capable of making progress from any non-optimal state. 
	\begin{assumption}[Unique best problem solver]\label{as:unique best agent}
		$|\arg\max_{\phi\in\Phi}\{\mathbb{E}_\nu(V\circ\phi)\}|=1$, i.e., there is only one best-performing agent according to \eqref{eq:expected value}.
	\end{assumption}
In other words, this assumption states that there is only one problem solver that performs best on average. There is only one problem solver that is the most likely to lead to the optimal state.
	\subsection{Problem solver interaction assumptions}
	Here, we turn to the interaction between problem solvers.
	\begin{assumption}[In series deliberation]\label{as:group del}
	The problem solver or agents $\Phi'\coloneqq\{\phi_1,\ldots, \phi_N\}$, when working together to solve the problem starting at $x$ are equivalent to the following sequence:
	\begin{enumerate}
	\item First, a problem solver $i_1$ such that  $x_1\coloneqq\phi_{i_{1}^x}(x)\neq x_0\coloneqq x$.
	\item Second, a problem solver $i_2$ such that  $x_2\coloneqq\phi_{i_{2}^x}(x_2)\neq x_1$.
	\item Inductively, a problem solver $i_j$ such that  $x_j\coloneqq\phi_{i_{j}^x}(x_{j-1})\neq x_{j-1}.$
\end{enumerate}
This stops at $x_n$, returning that value, such that it is a fixed point for all elements of  $\{\phi_1,\ldots, \phi_N\}$ (all the agents are stuck at the same point, unanimity).
\end{assumption}
Several remarks are in order. There can be multiple sequences arriving at the same point. The fixed point exists as $x^*$ is a fixed point for all elements of $\Phi$ by assumption. The group performance is tantamount to composition of the functions in a proper way:
	$$
	\phi^{\Phi'}(x)\coloneqq\phi_{i_n^x}\circ\ldots \phi_{i_1^x}(x)\,.
	$$
In other words, this assumption states that a group of problem solvers can be thought of as a sequence of agents that takes turns applying the problem solvers in the group to the current state, such that we approach the optimal value. The group will stop when it reaches a state that is a fixed point for all of the problem solvers in the group, i.e., unanimity.
\begin{assumption}[Clones]\label{as:clones}
	There exists an infinite amount of identical copies of each agent $\phi\in\Phi$.
\end{assumption}
This assumption implies that an infinite number of each problem solver is available. This condition is necessary for the theorem to hold, as will be discussed in Section \ref{sec:clone performance} and \ref{sec:clone selection}. Essentially, when the probability apparatus is introduced by Hong and Page, the selection criterion implies that an agent can be chosen multiple times, with repetition, indicating that multiple copies of each agent should be available.
\subsection{An example}
To give an example, let \(X=\{a,b,c,d\}\) and \(\Phi:=\{\phi_1,\phi_2,\phi_3\}\)  such that
\begin{center}
	\begin{tabular}{c|c|c|c|c}
		$x$& $V (x)$ & $\phi_1 (x)$ & $\phi_2 (x)$ & $\phi_3 (x)$ \\
		\hline
		$a$ & $1/4$ & $b$ & $a$ & $b$ 	\\
		$b$ & $1/2$ & $b$ & $c$ & $b$ 	\\
		$c$ & $3/4$ & $d$ & $c$ & $c$ 	\\
		$d$ & 1 	 & $d$ & $d$ & $d$ 			\\
	\end{tabular}.
\end{center}
For instance, this means that \(V(a)=\frac{1}{4}\) or \(\phi_2(b)=c\). Then, the hypotheses of the theorem are satisfied. Indeed, the problem is:
\begin{enumerate}
	\item \textbf{Problem}. The task is to maximize the value \(V(x)\). Looking at the table, the objective is to find \(d\), which has a value of 1.
	\item \textbf{Problem solver}. We have three problem solvers, \(\phi_1\), \(\phi_2\), and \(\phi_3\). These solvers transform an initial state \(x\) (one of the values \(a, b, c, d\)) into another state, which could be the same or different.
\end{enumerate}
The hypotheses:
\begin{enumerate}
	\item \textbf{Unique problem evaluation}. All problem solvers use the same value function \(V(x)\) to assess the desirability of a solution.
	\item \textbf{Unique solution}. There is a singular best solution, which is \(d\), having a value \(V(d)\) equal to 1.
	\item \textbf{Strictly increasing problem}. The values in \(V(x)\) are distinct and can be ranked in increasing order: \(V(a)=\frac{1}{4} < V(b) =\frac{1}{2} < V(c)=\frac{3}{4} < V(d)=1\).
	\item \textbf{Everywhere ability in problem solvers}. When applied to any state, a problem solver either maintains the value or improves it. For instance, \(\phi_1\) strictly increases the value of the states \(\{a, c\}\) and maintains the value of \(\{b, d\}\).
	\item \textbf{No improvement, idempotence}. If you apply a problem solver to a state repeatedly, the outcome remains the same. For example, applying \(\phi_1\) to \(a\) gives \(b\) and, if applied again, gives \(b\) too.
	\item \textbf{"Difficulty", imperfect problem solvers}. Each solver has some limitations. For instance, neither \(\phi_1\) nor \(\phi_3\) can transform \(b\) to \(d\).
	\item \textbf{"Diversity", sufficient unstuck problem solvers}. For every non-optimal state, at least one problem solver can make progress. For example, for state \(b\), \(\phi_2\) can improve it to \(c\).
	\item \textbf{Unique best problem solver}. Let's compute the expected value for each problem solver using the uniform distribution over \(\{a, b, c, d\}\). As we see, \(\phi_1\) is the best problem solver. 
		\item For $\phi_1$:
		\begin{itemize}
			\item $\mathbb{E}(V \circ \phi_1) = \frac{1}{4} \cdot (V(b) + V(b) + V(d) + V(d)) = \frac{1}{4} \left(\frac{1}{2} + \frac{1}{2} + 1 + 1\right) = \frac{3}{4}$.
			\item $\mathbb{E}(V \circ \phi_2) = \frac{1}{4} \cdot (V(a) + V(c) + V(c) + V(d)) = \frac{1}{4} \left(\frac{1}{4} + \frac{3}{4} + \frac{3}{4} + 1\right) = \frac{11}{16}$.
			\item $\mathbb{E}(V \circ \phi_3) = \frac{1}{4} \cdot (V(b) + V(b) + V(c) + V(d)) = \frac{1}{4} \left(\frac{1}{2} + \frac{1}{2} + \frac{3}{4} + 1\right) = \frac{11}{16}$.
		\end{itemize}
	\item \textbf{In series deliberation}. If solvers work together in sequence, they iteratively improve a state. Starting with \(a\), if we first apply \(\phi_1\) and then \(\phi_2\), we reach the state \(c\). If we apply \(\phi_1\) again, we reach \(d\).
	\item \textbf{Clones}. Imagine having infinite copies of each \(\phi\). This would mean we can always use any solver as needed, even if multiple instances of the same solver are required. As we said, when the probability apparatus is introduced by Hong and Page, the selection criterion implies that an agent can be chosen multiple times, with repetition, indicating that multiple copies of each agent should be available.
\end{enumerate}

\subsection{Trivial corollaries from the assumptions}
By construction (i.e., from the assumptions and \textit{not} the theorem), we have the following corollaries. Note that no profound or even standard mathematical results are needed; we just need the assumptions mentioned above combined with trivial arithmetic and trivial properties of sets.
\begin{corollary}\label{cor:diverse group}
	All members of $\Phi$ working together can solve the Problem $\forall~x\in X$. 
\end{corollary} 
\begin{proof}
This follows straightforwardly from Assumptions \ref{as:injective}, \ref{as:ability}, \ref{as:diversity}, and \ref{as:group del} such that,
	$$
	V(x)<V\left(\phi_{i_{1}^x}(x)\right)<\ldots <V\left(\phi_{i_{n}^x}(x_{n-1})\right)=1\,,
	$$
	for some $n\le|X|$. Indeed, by Assumption \ref{as:diversity} there is an agent which transforms the previous state, starting at $x$ and assuming it is not $x^*$, to a different one. By Assumption \ref{as:injective} and \ref{as:ability}, the value is strictly superior to the one of the previous state. This can be done till they reach $x^*$.
\end{proof}
\begin{corollary}\label{cor:best agent}
	There exists a state $~x\in X$ such that $\forall N_1\in\mathbb{N}$, $N_1$ ``clones'' of the the best performing agent cannot solve the Problem.
\end{corollary} 
\begin{proof}
	By Assumptions \ref{as:idempot} and  \ref{as:group del}, $N_1$ ``clones'' of the best performing agent work in the same way as a single clone alone. By Assumption \ref{as:difficult}, no agent can solve the problem alone for some $x$, so the corollary follows straightforwardly.
\end{proof}
So, just ``rearranging" the assumptions we can trivially (again, no profound mathematics, just basic arithmetic and trivial properties of sets) prove the following version of the ``theorem'' containing the main conclusion. The advantage of this formulation is that no clones are needed, compare with Theorem \ref{th:HP PNAS}.
\begin{theorem}[Basic Hong-Page theorem] \label{th:basic HP} Given the Assumptions \ref{as:unique prob}-\ref{as:group del}, all problem solvers (or a selected subset of them) working together perform better than the best problem solver (in the sense that there is a state $x$ such that the best problem solver cannot solve but the whole group can). Note that the best problem solver can be included in the first group.
\end{theorem}
\begin{proof}
This follows straightforwardly from the assumptions outlined above. Specifically, a proper subset of $\Phi$, as described in Corollary \ref{cor:diverse group}, can always solve the problem, whereas $\phi^*$ cannot, as indicated in Corollary \ref{cor:best agent}. Note that, according to Assumption \ref{as:ability}, including more agents does not worsen the state in terms of its value; the state either improves or remains the same.
\end{proof}
Again, this result is a simple corollary given the way we have formulated our hypotheses. Let us explain the proof in words. By assumption, we have arranged for the best agent not to always solve the problem. On the other hand, by mere assumption, for every state, there is an agent that can reach a better state. As the number of states is finite, they will reach the optimum in a finite number of steps. Note that the ``diverse group'' includes the best agent. The fact that the ``diverse" group'' outperforms the best agent is trivial, as they include additional agents and, by hypothesis, they do not worsen the solution and improve it for some states. 

Actually,  we do not need the whole group $\Phi$ to outperform the best performing agent. As $\mathbb{E}\left(V; \Phi\right)=1$, it is evident that there exists a $\tilde{\Phi}_0\subseteq \Phi$ such that
\begin{equation}\label{eq:tilde Phi def}
	\mathbb{E}\left(V; \tilde{\Phi}_0\right)>\mathbb{E}\left(V; \{\phi^*\}\right)<1\,.
\end{equation} 
That is the ``selected group'' of the statement of Theorem \ref{th:basic HP}. To connect with the original statement of Hong and Page, let us formulate the following version\footnote{See Remark \ref{rem:Phi0} for several comments on different definitions of $\Phi_R$.}.
\begin{theorem}[Deterministic Hong-Page's Theorem]\label{th:deterministic HP}
Given the hypotheses above and two natural numbers $0<N_1\le N$:
\begin{enumerate}
	\item[$\left(\mathcal{A}_1\right)$] \textbf{"Ability-diversity" assumptions}. The assumptions given above, Assumption \ref{as:unique prob}-\ref{as:group del}.
	\item[$\left(\mathcal{A}_2\right)$:] \textbf{"Counting clones" assumptions}. Assumption \ref{as:clones} and we choose two groups from a selected pool with $N$ clones of $\Phi$ that makes the following possible. In the first group there are $N_1$  clones such that $\tilde{\Phi}_0\subset\{\phi^R_1,\ldots,\phi^R_{N_1}\}=:\Phi_R$, and, in the second, there are $N_1$ clones of the best performing agent, $\Phi_B$.
\end{enumerate}
Then, the performance of $\Phi_R$ is better than $\Phi_B$ in the sense of \eqref{eq:expected value}.
\end{theorem}
\begin{proof}
	It is trivial by the corollaries given above. Indeed, by the first corollary and \eqref{eq:tilde Phi def}$, 		\mathbb{E}_\nu(V\circ{\phi^{\Phi_R}})>\mathbb{E}_\nu(V\circ{\phi^{\Phi_B}})\,.$ Note that by the second corollary, $\mathbb{E}_\nu(V\circ{\phi^{\Phi_B}})<1$ as $\exists_{\ge 1}~x$ such that $V(\phi^{\Phi_B}(x))<1.$
\end{proof}
In other words, as detailed in Appendix \ref{ap:simpler proof}, the theorem has the following structure with an evident proof:

\vspace{1.7mm}
\noindent\textsc{Theorem}. \textit{{Hypotheses}}: Let there be a pool of agents such that:
\begin{itemize}
	\item they never degrade the solution and where there is always an agent who can improve the solution,
	\item the best performing agent is imperfect by hypothesis,
	\item as it follows from these two points above, we can choose agents randomly until they outperform the imperfect best performing agent.
\end{itemize}
{\textit{Thesis}}: Consequently, this ``random" group will outperform the best performing agent.

So, what did Hong and Page prove in their article? Essentially, they demonstrated that assumption $\mathcal{A}_2$ holds almost surely (after defining some probability measures). See their proof of Lemma 1 and Theorem 1 in \cite{HP04}. This is a not very difficult probabilistic claim that has nothing to do with either ability or diversity, which are contained in the assumptions $\mathcal{A}_1$. It is a probabilistic fact that can be shown, regardless of whether the objects considered are diverse agents, incapable problem solvers, balls in a box, or mathematical functions in a Hilbert space\footnote{That is, the probability apparatus is introduced to ensure that with probability one, we have the groups described in $\mathcal{A}_2$, as further discussed in Section \ref{sec:clone selection}. This approach could be applied to any object, unrelated to diversity, and would be simplified by directly stating $\mathcal{A}_2$ or, even better, presenting a clone-free version of it, as in Proposition \ref{th:basic HP}. In Section \ref{sec:ability trumps div}, detailed in Appendix \ref{ap:ability trumps div th}, we will employ probability in a manner pertinent to an actual deliberative process.}. Nevertheless, this is the heart of proof of their article published in the Proceeding of the National Academy of Sciences. But one might question, given that we have shown that Theorem \ref{th:basic HP} (a trivial restatement of the assumptions) encapsulates all the information regarding diversity and ability, what is the necessity of introducing clones? This is why Thompson say that the theorem ``is trivial\footnote{See Remark \ref{rem:CJTvsHP} on why the same cannot be said of other theorems, like the Condorcet Jury Theorem.}. It is stated in a way which obscures its meaning. It has no mathematical interest and little content." We can compare the previous version, Theorem \ref{th:deterministic HP}, with the original statement:
\begin{theorem}[\cite{HP04}]\label{th:HP PNAS}
	Given the assumption above, let $\mu$ be a probability distribution over $\Phi$ with full support. Then, with probability one, a sample path will have the following property: there exist positive integers $N$ and $N_1$, $N > N_1$, such that the joint performance of the $N_1$ independently drawn problem solvers exceeds the joint performance of the $N_1$ individually best problem solvers among the group of $N$ agents independently drawn from $\Phi$ according to $\mu$.
\end{theorem}
As noted by Thompson, \cite{Tho14}, the theorem as originally stated was false because Assumption \ref{as:injective} was not included. See Section \ref{sec:V inj} for more details. Note that ``$N_1$ individually best problem solvers'' are just clones of the best problem solver (unique by assumption), not, for instance, the first and second according to their expected value (which will perform better). This restriction is imposed by hypothesis.
More precisely, we need a new assumption about the selection criteria for the members of the comparison groups employed in the argument.
\begin{assumption}\label{as:random}
By hypothesis we have the following.
	\begin{enumerate}
		\item The first group is selected randomly from a pool of clones of the elements in $\Phi$. The group size $N_1$ can be adjusted as required.
		\item  Similarly, the second group is chosen independently, but from an identically distributed set of clones of the elements in $\Phi$ of size $N_1$. This selection process follows the stipulations that:
		\begin{itemize}
			\item the group size $N$ can be adjusted as required,
			\item the selection allows for the repetition of the best problem solvers.
		\end{itemize}
	\end{enumerate}
	
\end{assumption}

\subsection{Other results and simpler proof}
\begin{proposition}\label{prop:other results}
	Assuming the conditions of Theorem \textnormal{\ref{th:HP PNAS}} with $N_1$ large enough, with probability one,
	\begin{enumerate}
		\item the randomly selected group of $N_1$ problem solvers will invariably converge on the correct solution without any disagreement and unanimity, 
		\item the ``random group'' always contains the best-performing agent.  
	\end{enumerate}
These facts explain that they can always outperform the best problem solvers.	 
\end{proposition}

\begin{proof}
	This is straightforward from Corollary \ref{cor:diverse group} and that, following the first part of Assumption \ref{as:random}, the first group includes a copy of $\Phi$ or $\Phi_0$ (Remark \ref{rem:Phi0}) $\mu$- almost surely. It is also Lemma 1 in \cite{HP04}. There is unanimity as, for every state $x\in X$, the group solution is $x^*$, where everyone accepts as a solution, $\phi(x^*)=x^*$. The second statement follows from the Strong Law of Large Numbers, see Remark \ref{rem:N1 large} for more details.
\end{proof}
For instance, the following $\Phi\coloneqq\{\phi_1,\phi_2,\phi_3\}$ such that
\begin{center}
	\begin{tabular}{c|c|c|c|c}
		$x$ & $V (x)$ & $\phi_1 (x)$ & $\phi_2 (x)$ & $\phi_3 (x)$ \\
		\hline
		$a$ & $1/4$ & $b$ & $a$ & $b$ 	\\
		$b$ & $1/2$ & $b$ & $c$ & $b$ 	\\
		$c$ & $3/4$ & $d$ & $c$ & $c$ 	\\
		$d$ & 1 	 & $d$ & $d$ & $d$ 			\\
	\end{tabular}
\end{center}
satisfies the hypotheses of the theorem, but, if the ``random'' group does not include the best performing agent\footnote{Assumptions can be made to exclude the best performing agent, while ensuring that there is another agent that performs as the best one does when needed. Consequently, it is no surprise that the theorem still holds. However, this approach is purely ad hoc.}, $\phi_1$, then it cannot outperform $\phi_1$. 

That is, given the same hypotheses, one could simply select the $N_1$ of Proposition \ref{prop:other results} and assert a stronger result. It is not just that the random group outperforms the best
agent, but they also almost surely find the optimal solution. So, with the same hypotheses,
one can prove a stronger result. Shouldn’t one then fully explore the implications of these
hypotheses and reject them if they lead to unrealistic conclusions? That is, given that $N_1=N_1(\omega)$ is \textit{not fixed} but rather a ``stopping time" that halts when the desired members, as per the selector's or chooser's preference, are included, one might ask why the selector should stop merely to outperform the best agent if the correct solution can always be reached when the hypotheses are met. This more comprehensive result trivially implies the one in Theorem \ref{th:HP PNAS}. See Appendix \ref{ap:simpler proof} for more details.

\subsubsection{{Triviality of the result: simpler proof.}}
\label{sec:simpler proof}
In fact, a simpler proof of the theorem can be constructed based on that simple fact. This approach also exposes \textit{the theorem's triviality given its underlying assumptions}. See Appendix \ref{ap:simpler proof} for the technical details. An intuitive explanation goes as follows:

The first paragraph of the proof in Appendix \ref{ap:simpler proof} corresponds to the part of the theorem where diversity and ability are put into play. This essentially reduces to the following triviality: by assumption, there are two distinct agents - the best agent and another agent - and a state $x_*$ for which the best agent does not provide the optimal solution. However, the other agent can improve upon the solution of the best agent for this state. This implies that the performance of a group consisting of the best agent and this additional agent surpasses the performance of the best agent alone, at least for some states. For other states, again by assumption, adding an agent does not worsen the situation, thus completing the deterministic clone-free part of the proof. Subsequently, we apply the strong law of large numbers to ensure that, under the setting of Assumption \ref{as:random}, the random group will always contain copies of these two agents, and the best-performing agents are all copies of the unique best-performing agent.

\section{Removing technical hypotheses: Counterexamples}\label{sec:3}
The theorem depends critically on certain assumptions that we are going to analyze now. In this section, we will refrain from critiquing certain empirical hypotheses, such as the assumption that agents share the same concept of problem-solving (Assumption \ref{as:unique prob}, although see Remark \ref{rem:Rawls}), or that they can recognize the solution ($\phi(x^*)=x^*$). Such critiques largely pertain to the plausibility inherent in every model, and one could always defend these by invoking ideal conditions, much as one might assume frictionless systems in physics. Although these critiques can be adequate, a different critique, following a ``Moorean style," will be presented in Section \ref{sec:ability trumps div} and Section \ref{sec:Land hyp}. There, we will revisit some empirical hypotheses (not the ones mentioned above), slightly modifying them to enhance their plausibility, which may lead to contrary conclusions. We will also see how those hypotheses lead to contradictions. However, in this section, we wish to focus on certain technical assumptions, often overlooked\footnote{For instance, see Section \ref{sec:rem book hypotheses}.}, that are essential for the theorem to hold. Without these assumptions, the theorem fails, giving our counterexamples. These technical assumptions, by their nature, involve facets of the model (not the underlying reality) that are difficult to verify, hence making it challenging to argue for their plausibility. This raises the question of why we should adopt these hypotheses, rather than others.
\begin{remark}
The distinction between empirical and technical assumptions might seem somewhat arbitrary, but it nonetheless serves a useful purpose in our analysis. For instance, if we apply the theorem to a jury in a criminal trial, how do we model clones of the jurors and how do we select them from (infinite) groups? Similarly, as we will see, the values of \(V\) (apart from \(V(x^*)=1\), the right option) are important for the theorem to hold; if certain conditions are not met, then the thesis of the theorem fails. However, how could one verify that the hypotheses on \(V\) hold when \(V\) is not empirically observable? What is the value or how to determine it for the solution of being charged with different combinations of crimes while being innocent? Recall also Assumption \ref{as:unique prob}, \(V_\phi\) could be more easily observed, but we need the ``social'' function $V$.
\end{remark}

\subsection{$V$ is an injection, Assumption \ref{as:injective}}\label{sec:V inj}
This was pointed out by Thompson and we reproduce it here with minor modifications. This assumption was not originally in \cite{HP98, HP04}, making the theorem false.

Let $X = \{a, b, c, d\}$. Define $V (x)$ and three agents $\phi_1$, $\phi_2$ and $\phi_3$ according to the table below:

\begin{center}
	\begin{tabular}{c|c|c|c|c}
	$x$ & $V (x)$ & $\phi_1 (x)$ & $\phi_2 (x)$ & $\phi_3 (x)$ \\
	\hline
	$a$ & 1/3 & $d$ & $c$ & $b$ \\
	$b$ & 2/3 & $b$ & $c$ & $b$ \\
	$c$ & 2/3 & $c$ & $c$ & $b$ \\
	$d$ & 1 & $d$ & $d$ & $d$ \\
\end{tabular}

\end{center}
The set of agents $\Phi = \{\phi_1, \phi_2, \phi_3\}$ satisfies all
the hypotheses of the theorem. The agents $\phi_1, \phi_2, \phi_3$
have average values $5/6, 9/12, 9/12$ respectively, so
$\phi_1$ is the “best” agent. Notice that all three agents
acting together do not always return the point $d$,
where the maximum of $V$ occurs. Indeed all three agents
acting together work only as well as $\phi_1$
acting alone. Hence in this case, no group of agents
can outperform $\phi_1$, or, equivalently, multiple
copies of $\phi_1$, hence no $N$ and $N_1$ exist which
satisfy the theorem. See Remark \ref{rem:inj reply} on why some replies given to Thompson on this regard are not satisfactory.

\subsection{Unique best agent, Assumption \ref{as:unique best agent}}\label{sec:unique best}
To justify this assumption, Hong and Page write:
\begin{quote}
Let $\nu$ be the uniform distribution. If the value function $V$ is one to one, then the uniqueness assumption is satisfied. 
\end{quote}
This mathematical statement is false. Indeed, let us consider $X = \{a, b, c, d\}$. Define $V (x)$ such that $0<V(a)<V(b)<V(c)<1$, $V(b)<\frac12\left(V(a)+1\right)$ and $n$ agents $\phi_1$, $\phi_2$ and $\phi_i$ according to the table below:

\begin{center}
	\begin{tabular}{c|c|c|c|c}
	$x$ & $V (x)$ & $\phi_1 (x)$ & $\phi_2 (x)$ & $\phi_i (x)$ \\
	\hline
	$a$ & $V(a)$ & $a$ & $c$ & $\phi_i(a)$ 	\\
	$b$ & $V(b)$ & $b$ & $c$ & $\phi_i(b)$ 	\\
	$c$ & $V(c)$ & $d$ & $c$ & $\phi_i(c)$ 	\\
	$d$ & 1 	 & $d$ & $d$ & $d$ 			\\
\end{tabular}
\end{center}

The set of agents $\Phi = \{\phi_1, \phi_2, \phi_i\}_{i=3}^N$ satisfies all
the hypotheses of the theorem and are ordered according to its expected value. If we set
\begin{equation*}
	V(c)\coloneqq \frac13\left(V(a)+V(b)+1\right)\,,
\end{equation*}
then $\phi_1, \phi_2$ have the same ``expected ability'' under the uniform measure. Furthermore, now the theorem is false. Indeed,
$$
\phi_1\circ \phi_2(a)=d,~~\phi_1\circ\phi_2(b)=d,~~\phi_1(c)=d,~~\phi_1(d)=\phi_2(d)=d\,.
$$
In this case, no group of agents
can outperform $\{\phi_1, \phi_2\}$, no $N$ and $N_1$ exist which satisfy the theorem.
\begin{remark}
Here, we have demonstrated an example involving two agents possessing identical ``expected abilities". Of course, in real-world applications, there would likely be uncertainty or variability in the value of $\mathbb{E}_\nu\left(V\circ\phi\right)$; thus, it would be prudent to consider an interval rather than a single point. In such circumstances, the top-performing agents might comprise multiple individuals with high probability. However, as demonstrated, the theorem may not necessarily hold in these scenarios.
\end{remark}
\subsection{Clones performance}\label{sec:clone performance}
As we saw, simply by the assumptions, one million Einsteins, Gausses or von Neumanns are the same as just one of them. Indeed, mathematically, by Assumption \ref{as:clones}, $\{\phi,\ldots,\phi\}$ is a well-defined set of problems solvers such that
\begin{equation*}
	\phi^{\{\phi,\ldots,\phi\}}\stackrel{\text{Assumption }\ref{as:group del}}{=}\phi\circ\ldots\circ\phi\stackrel{\text{Assumption }\ref{as:idempot}}{=}\phi.
\end{equation*}
In other words, by the way the deliberation is structured and the idempotency assumption, a million clones of the same agent are equivalent to one. 
Again, those are just the assumptions they arbitrarily made. But this may not make much sense if we want to apply it to real-life scenarios. More realistic versions could be:
\begin{itemize}
	\item \textit{Improvement}: $V\circ\phi\circ\phi \ge V\circ\phi$ (strict inequality for some points). In other words, if an agent (competent) produced a solution after a certain amount of time, say one hour, it would provide a better answer if it had one-million hours, or if a ``clone'' could pick up where he left off.
	\item \textit{Work in parallel\footnote{As a technical note, now $\{\phi, \phi\}$ should be considered a multiset (the multiplicity distinguishes multisets).}}: $V\circ \phi^{\{\phi, \phi\}}\ge V\circ \phi$ (strict inequality for some point). In other words, one can imagine that a group of Einsteins would not work sequentially, always producing the same result, but would divide the work, resources, focus, etc. to produce a better answer once they have put all of their findings together.
\end{itemize}
It is not clear which is the best way to treat formally the performance of clones in their theorem. But as the previous arguments show, even in conditions of ideal theorizing, there are more reasonable assumptions that could have been chosen. However, these alternative assumptions do not guarantee the conditions required for the theorem to hold. Otherwise, as $N_1\to\infty$, no group of agents could generally outperform ${\phi,\stackrel{N_1}{\ldots},\phi}$; we cannot guarantee the existence of an $N$ that would satisfy the theorem.
\begin{remark}
Following Jason Brennan's recurrent ``magic wand'' thought experiment, let's imagine we are confronted with an exceedingly difficult problem to solve, for instance, the Navier-Stokes Millennium Problem. Suppose we have a magic wand at our disposal that can create agents to solve the problem for us. Should we choose Terence Tao, or should we use the magic wand to create 100 Terence Taos working together to solve our problem? According to the assumption of the Hong-Page Theorem, this magic wand would be useless. 
\end{remark}
See Remark \ref{rem:X inf} for why the critique still holds in the case of infinite $X$.

\subsection{Selection of clones}\label{sec:clone selection} Similarly, the assumptions for clone selection appear to be somehow arbitrary, and it is not justified why these assumptions, that facilitate reaching the conclusion that ``diversity trumps ability'' should be chosen over other perfectly reasonable and simpler ones:

\begin{itemize}
	\item The choice of two independent groups seems arbitrary. Why not fix $N$ and, from the same group, select a random subgroup of size $N_1$, as well as the best $N_1$ problem solvers, and then compare? In such a scenario, the theorem might not hold. Indeed, we need $N\gg N_1$ such that the Strong Law of Large Numbers (SLLN) applies, $\mu\left(\{\phi^*\}\right)$ can be very small. However, a random group of $N_1$, $\Phi_{N_1}$, agents might not include all the needed problem solvers of $\Phi$, thus we cannot guarantee a probability of one, as the theorem does. That is, for $N>N_1$, it could be the case that, see Remark \ref{rem:Phi0},
	\begin{equation*}
		\mathbb{P}\left(\tilde{\Phi}_0\subset \Phi_{N_1}\right)<1\,.
	\end{equation*}
	\item Permitting repetition is also arbitrary. We could, for instance, select the best problem solvers without allowing repetitions, as is standard in probability problems. Recall from Section \ref{sec:clone performance} that adding a repeated clone is equivalent to adding nothing. This could prevent the paradoxical result that, by mere hypotheses, when choosing the best problem solvers from a group of size $N$ is more beneficial when the group size is relatively small, i.e., for choosing the best it is preferable to have less options available. However, if we prohibit repetitions, then the theorem does not necessarily hold as the best problem solvers can include the ones (without taking repeated clones into account) of the ``random" group, so no $N$ and $N_1$ exist which 	satisfy the theorem.
\end{itemize}
We should note the general approach adopted by Hong and Page, see also the description of the infinite pool of clones in Section \ref{sec:rem book hypotheses}.
\begin{enumerate}
	\item They introduce randomness into their model by assuming an infinite number of clones for each agent.
	\item They invoke the SLLN (or similar results) to ensure that the frequency of appearance converges to the original probability, given by $\mu$.
	\item As the group sizes involved in the SLLN could be large, $N_1, N$ are not fixed, but also random. For instance, $N_1=N_1(\omega)$ is chosen such that the ``random" group contains enough agents to outperform $\phi^*$. This cancels out the randomness previously introduced.
\end{enumerate}
All in all, the steps detailed above effectively eliminate the randomness that was introduced, transforming the framework into a \textit{de facto} deterministic model. That is, after the ``complex'' probabilistic steps, with probability one, the groups being compared are, essentially\footnote{See Appendix \ref{ap:simpler proof} for more details.}, the ones in Theorem \ref{th:basic HP}. But, as we saw, if the results were presented this way, the triviality would be manifest. In the next section, we will remove clones and explore the consequences of a fair comparison of ability and diversity in the Hong-Page setting.

\section{New Hong-Page style theorem: Ability trumps diversity}
\label{sec:ability trumps div}
The objective of this section is to make a fairer comparison of diversity and ability while maintaining the Hong-Page framework.  We are going to state and prove a new version of the Hong-Page theorem such that the hypotheses are going to be of the same kind and as plausible (or even more as we will see, for instance, no need of clones and disagreement is possible) as the ones in the Hong-Page theorem. Nevertheless, we will reach the opposite conclusion, ``ability trumps diversity". We are not claiming that this theorem has any social content; it simply shows how the justification of the theorem is driven by the selection of assumptions that might be overly strict or arbitrary.

The moral would be if we create two groups from the group in the original theorem – one in which we make the minimal reduction in ability while ensuring full diversity, and another in which we considerably reduce diversity while ensuring ability – the less diverse group would systematically outperform the fully diverse group. In other words, ability trumps diversity. Here we present a non-technical exposition, see Appendix \ref{ap:ability trumps div th} for all the mathematical details.

\subsection{The new assumptions}
We start with Hong-Page's framework, but redefining some of their assumptions to incorporate the potential for disagreement within a group of agents, introducing a realistic element into our model. We describe a situation where one agent's suggestion can be directly countered by another, leading to a deadlock, symbolizing a disagreement.
Unlike the original Hong-Page Theorem's assumptions, where agreement is always achieved, our model accepts the possibility of disagreement, reflecting a more lifelike scenario where unanimous decisions might not always be possible. Clones are not needed now, so idempotency is not necessary.

Also, we then introduce a probability measure to account for the likelihood of each agent influencing the decision process. This is done ensuring that every agent has a chance to contribute, emphasizing inclusivity in the decision-making process.

\subsubsection{The ability group}
We define an ability-focused group, $\Phi^\textnormal{A}$, with the following characteristics:
\begin{itemize}
	\item \textbf{Ability:} Every agent in this group is capable of maintaining or improving the current state, ensuring that their contributions are always constructive.
	\item \textbf{Common knowledge:} There exists a subset of states, $X_{\textnormal{CK}}$, where all agents in this group agree on the solution, showing non-diversity and competence.
	\item \textbf{Limited diversity:} While this group maintains some level of diversity, it is notably reduced outside of $X_{\textnormal{CK}}$, with a unique agent providing distinct solutions in these scenarios.
\end{itemize}

\subsubsection{The diversity group}
Conversely, the diversity group, $\Phi^\textnormal{D}$, is characterized by:
\begin{itemize}
	\item \textbf{Full diversity:}  its wide range of perspectives, ensuring that for almost every problem state, there's an agent capable of offering a different perspective.
	\item \textbf{Minimal ability loss:} this group includes the minimal loss of ability possible. However, this will serve to demonstrate that this is enough to show that diversity does not trump ability.
\end{itemize}

\subsection{The theorem}
\begin{theorem}[Ability Trumps Diversity Theorem]\label{th:ability main}
	Given the defined groups $\Phi^\textnormal{A}$ and $\Phi^\textnormal{D}$ with their respective properties, it is shown that the group selected for its ability rather than diversity performs better in reaching the optimal solution.
\end{theorem}

This theorem underscores the importance of selecting problem solvers not just for their diverse perspectives but for their ability to constructively contribute to solving complex issues, challenging the notion that diversity trumps ability. See Appendix \ref{ap:ability trumps div th} for a fully detailed exposition.

\section{Hong and Page's misuse of mathematics: a simple proposition with unnecessary complexity}\label{sec:6}

\subsection{Adding unnecessary complexity to a simple fact}\label{sec:obscuring}
The Hong-Page theorem is, in essence, a misuse of mathematics in the following sense: it employs probability techniques, such as the Borel-Cantelli lemma (as my simpler proof in Appendix \ref{ap:simpler proof} demonstrates, unnecessarily), to prove claims unrelated to diversity. This approach makes the theorem, about diversity and social intelligence, inaccessible to individuals outside the mathematical field. Essentially, it uses mathematics that complicate a simple fact rather than mathematics that simplify complex relations.

Indeed, as we saw in Theorem  \ref{th:basic HP} the theorem's conclusion—that a group performs better than the best single individual—is inevitable by construction, by the way the theorem's premises are structured. It posits two fundamental hypotheses: first, that the ``best" individual agent, $\phi^*$, cannot always solve the problem optimally, and second, that a diverse group $\Phi$ of agents can always find an optimal solution. When these assumptions are in play, the conclusion of the theorem (that the diverse group outperform the best agent) is logically guaranteed. 

But, although all related to diversity is there, clone existence and selection is introduced, to invoke the probability apparatus. This is done in the second set of assumption of Theorem \ref{th:deterministic HP}. They define a probability space and prove that, if we can select clones of $\Phi$ infinitely, with probability one the first group will contain at least one copy of the necessary elements of $\Phi$ and the second group will be chosen so that is made only of copies of $\phi^*$. Using the previous paragraph, the conclusion follows directly. This constitutes the heart of their article's proof published in the Proceedings of the National Academy of Sciences. But, since we've shown that Theorem \ref{th:basic HP}—a simple restatement of the assumptions—encapsulates all information regarding diversity and ability (the probabilistic part could be applied to anything, like colored stones in a box), one may question the necessity of introducing clones in the first place. Thus, it appears that the theorem's complexity may stem more from the unnecessary second step than from a deep, mathematical truth about diversity and ability, see Section \ref{sec:clone selection} for more details.

Thus, while the Hong-Page theorem uses mathematical techniques, its conclusion is more a trivial product of its constructed premises than a deep, unexpected and universal truth revealed through rigorous mathematical exploration, see also Remark \ref{rem:CJTvsHP}.

\subsection{Stretching the theorem to answer questions outside its range of applicability}\label{sec:answer}
In \cite{HP04}, they say:
\begin{emphasisquote}
These results still
leave open an important question: Can a functionally diverse
group whose members have less ability outperform a group of
people with high ability who may themselves be diverse? \textbf{The
main result of our paper addresses exactly this question.}
\end{emphasisquote}
They insist:
\begin{emphasisquote}
	To make a more informed decision, the organization administers a test to 1,000 applicants that is designed to reflect their individual abilities in solving such a problem. Suppose the applicants receive scores ranging from 60\% to 90\%, so that they are all individually capable. Should the organization hire (i) the person with the highest score, (ii) 20 people with the next 20 highest scores, or (iii) 20 people randomly selected from the applicant pool? Ignoring possible problems of communication within a group, the existing literature would suggest that ii is better than i, because more people will search a larger space, but says little about ii vs. iii. The intuition that agents with the highest scores are smarter suggests that the organization should hire ii, the individually best- performing agents. The intuition that the randomly selected agents will be functionally diverse suggests that the organization should hire iii, the randomly selected ones. \textbf{In this paper, we provide conditions under which iii is better than ii.}
\end{emphasisquote}

However, this is a mistake, incorrect. By Theorem \ref{th:deterministic HP}, the groups being compared consist of clones that include, at the very least, all agents necessary to consistently outperform the best agent. These are contrasted with clones of the best agent which, by assumption, are equivalent to the best agent alone. It's noteworthy that the best agent is generally included in the first group with a high probability, even one or 100$\%$, as indicated in Proposition \ref{prop:other results}. 
To put it another way, if we consider only one copy of each agent (because multiple copies are redundant by assumption, as discussed in Section \ref{sec:clone performance}), the groups being compared are almost\footnote{Or, more generally, a set \(\Phi_R\) that is a subset of \(\Phi\) and definitely includes \(\tilde{\Phi}_0 \text{ or } \Phi_0\), but might also comprise other elements of \(\Phi\). See Appendix \ref{ap:simpler proof} for further details. This distinction is largely inconsequential for the argument as elucidated in the appendix.} all agents, i.e., \(\Phi\) versus the autoimposed imperfect best performing agent, i.e., \(\{\phi^*\}\). It should be emphasized that \(\{\phi^*\}\) is a subset of \(\Phi\). There is no element of random selection involved, as elaborated in Section \ref{sec:clone selection}. Therefore, a more appropriate comparison would be:
\begin{enumerate}[label=\roman*]
	\item the person with the highest score, 
	\item 20 people with the next 20 highest scores, 
	\item 20 people randomly selected from the applicant pool,
	\item a subset of the 1000 applicants that contains, at a minimum (and possibly more), enough agents to outperform the best-performing agent.
\end{enumerate}
The Hong-Page paper deals with i) versus iv), a triviality, \textbf{not}, as they explicitly claim, ii) versus iii).

During a conference\footnote{The rationale for analyzing a presentation by Page at the ECB is twofold. Firstly, such presentations can potentially reach a broader audience than academic papers, thus significantly impacting public understanding of this theorem. It is therefore appropriate to examine whether there exists a discrepancy between the motte (actual theorem) and the bailey (public discourse). Secondly, these presentations offer insights into the author's simplified, yet ideally accurate, interpretations of his theorem, tailored for an informed and technical audience. Thus, analyzing it here provides valuable context and perspective for our discussion.} at the European Central Bank (ECB), Page\footnote{See also Section \ref{sec:rem book hypotheses} for more.} stated:

\begin{quote}
	I create a group of the 20 best agents - the best individuals - and I compare them to a random group of 20 agents [...] it turns out though if you do the math on this, the diverse group almost always outperforms the other group if you use reasonable-sized groups, like groups of size 10 or 20 [...] the paper and model I just showed you where diverse groups do better than random groups was written by myself and Lu Hong [...]
\end{quote}

He used Figure \ref{fig:pageecb} to illustrate this point. However, as we have mentioned before, this representation is not directly related to the theorem. In the ``Alpha Group", the best agent, 138, should be the only member, and this agent could also be included in the Diverse Group, along with all the other agents, see Figure \ref{fig:pageecb_actual}. Furthermore, groups of size 10 or 20 may not be large enough for the SLLN to hold, especially if $\mu\left(\lbrace\phi\rbrace\right)$ is small enough for some agents.
\begin{figure}
	\centering
	\begin{subfigure}{0.49\linewidth}
		\centering
		\includegraphics[width=\linewidth]{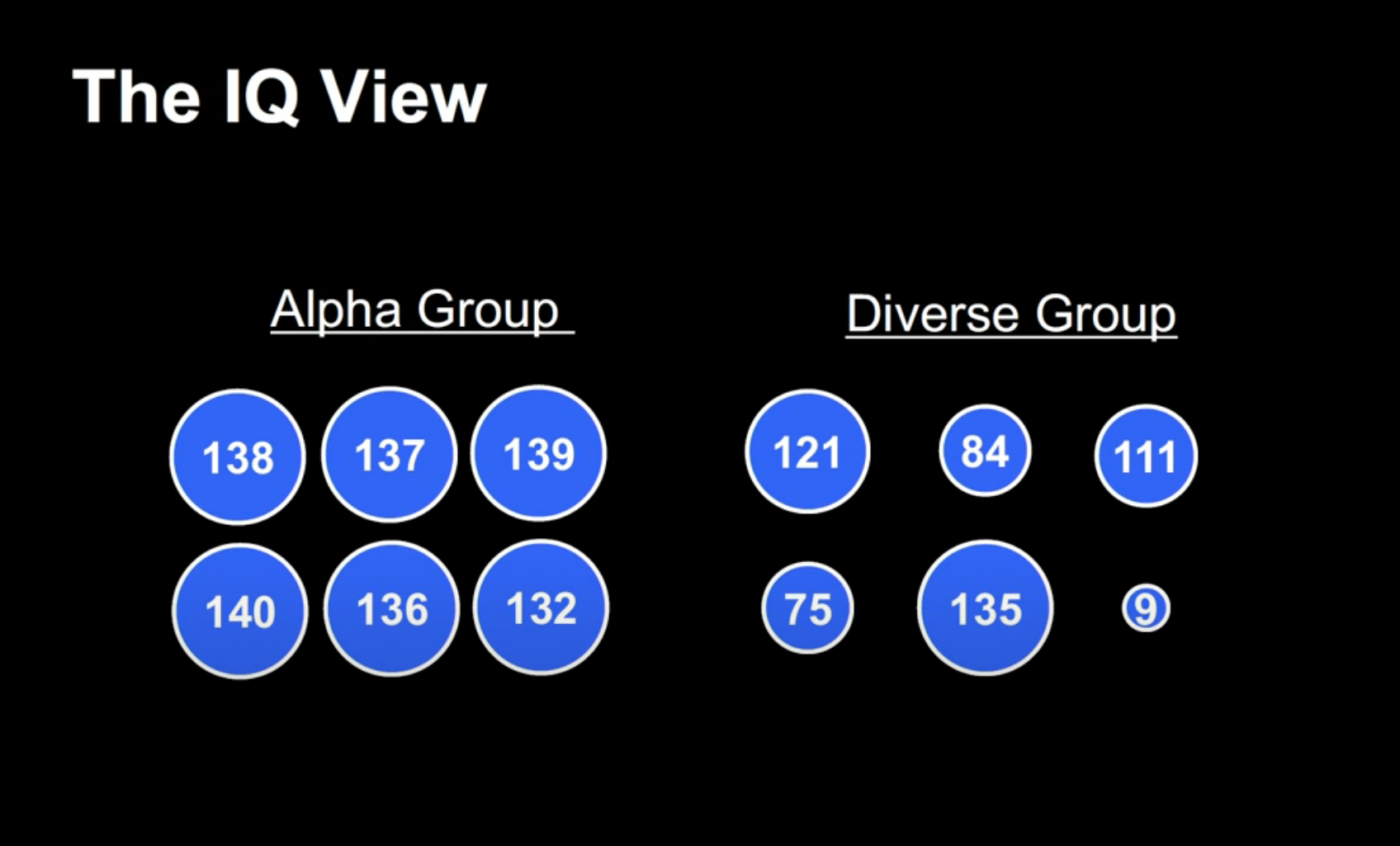}
		\caption{The slide used by Scott Page in his conference at the ECB.}
		\label{fig:pageecb}
	\end{subfigure}
	\hfill
	\begin{subfigure}{0.49\linewidth}
		\centering
		\includegraphics[width=\linewidth]{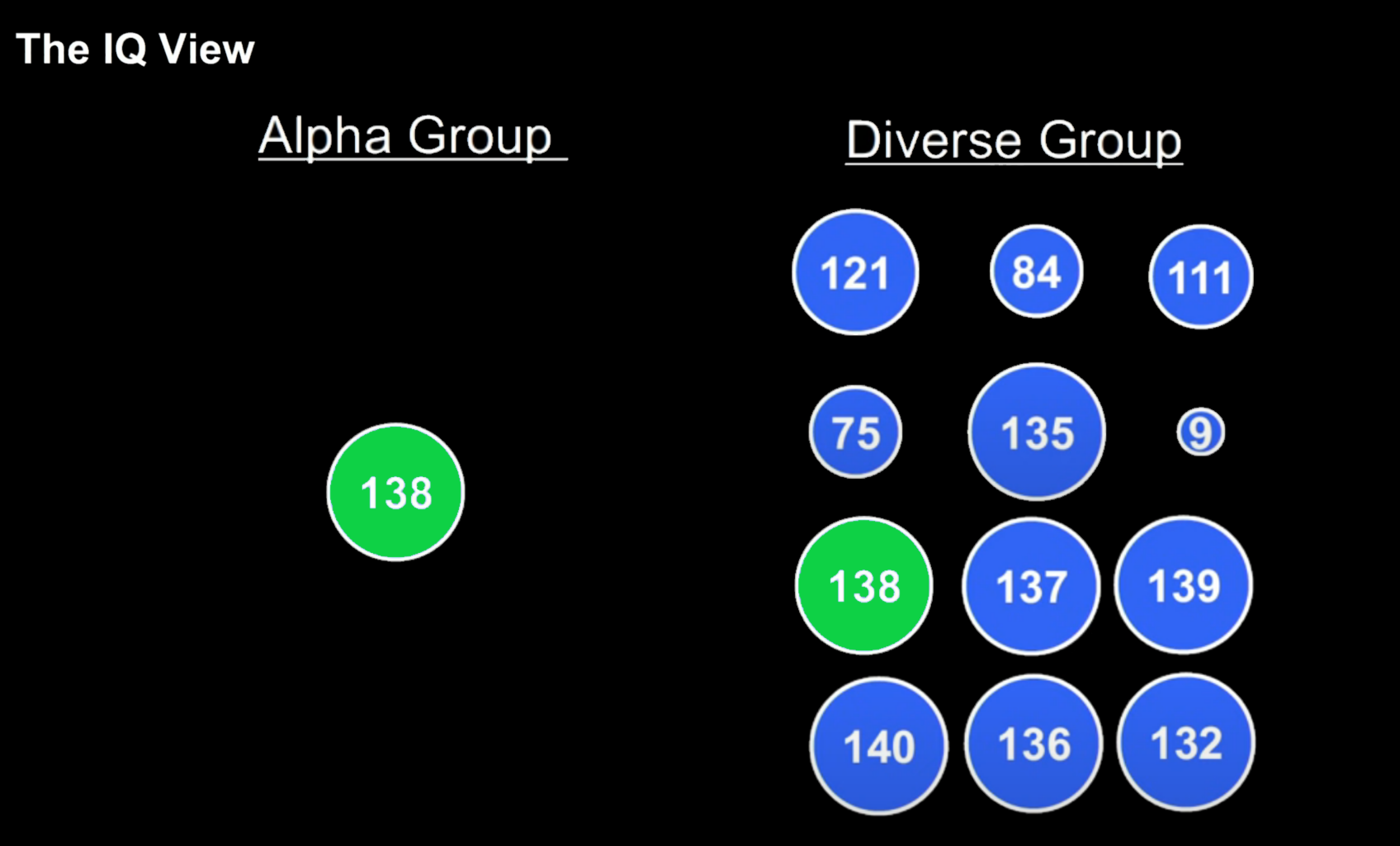}
		\caption{A presentation slide that better captures the essence of the theorem.}
		\label{fig:pageecb_actual}
	\end{subfigure}
	\caption{Comparison of the slides. The Hong-Page theorem ostensibly claims that the diverse group outperforms the alpha group. However, the groups presented in Page's slides do not represent the actual groups used for the theorem. The best agent (imperfect by their hypothesis) is shown in green. The group on the right, by their hypotheses, can always reach the correct solution. We are taking $N_1=N_1^{(1)}$ from Appendix \ref{ap:simpler proof}. For other possible $N_1$ and the corresponding figures, please refer to that appendix.}
	\label{fig:pageecb_comparison}
\end{figure}
In the same conference at the ECB, he further states:

\begin{quote}
	As the problem becomes complex, the best team doesn't consist of the best individuals. Why? Because the best individuals tend to be similar and what you really want on hard problems is diversity.
\end{quote}
Similarly, in \cite{Pag07}, he states:
\begin{emphasisquote}
	The reason diversity trumps ability is not deep: \textit{The best problem solvers likely have similar perspectives and heuristics}. The random problem solvers bring diverse ways of thinking. Therefore, the best problem solvers all	get stuck in the same places. The random problem 	solvers don't.
\end{emphasisquote}

However, these statements seems to confuse an assumption with a factual result. The claim that ``the best individuals are similar" (actually, clones of the same agent) is not a derived conclusion, but a trivial consequence of the presuppositions, the proof of this claim cannot be found in \cite{HP04}; it's established by assumption, Section \ref{sec:clone selection}. Furthermore, as explored in Section \ref{sec:ability trumps div}, even when conducting a fair comparison between ability and diversity - and even when the ability group is characterized by relatively homogeneous problem solvers - ability can still outperform diversity. Therefore, this statement appears to be misguided.
\subsection{Misusing the prestige of mathematics}
When Page claims \cite{Pag08}:

\begin{quote}
	This theorem ["Diversity Trumps Ability"] is no mere metaphor or cute empirical anecdote that may or may not be true ten years from now. It is a mathematical truth.
\end{quote}

It is as accurate as asserting:

\begin{quote}
	If $p\land q$, then $p\land q\,,$ is no mere metaphor or cute empirical anecdote that may or may not be true ten years from now. It is a mathematical truth.
\end{quote}
To be more precise, by Corollary \ref{cor:diverse group} and Proposition \ref{prop:other results}, the logical structure of the theorem is as follows:

\begin{enumerate}
	\item \textit{Hypothesis 1}, $H_1$: The best agent cannot always solve the problem.
	\item \textit{Hypothesis 2}, $H_2$: The ``diverse group'' can always solve the problem (or\footnote{See Remark \ref{rem:Phi0} and Appendix \ref{ap:simpler proof}.} outperform the best performing agent).
	\item \textit{Conclusion}, $C$: The ``diverse group'' outperforms the unique best agent at problem-solving, signifying that ``diversity trumps ability".
\end{enumerate}

This argument is logically valid—a tautology, so the proposition $H_1\land H_2\to C$ is certainly true. However, the argument's soundness might be questionable as the hypotheses might not be factual. Thus, it doesn't provide any certainty regarding the conclusion, $C$, i.e., whether diversity indeed outperforms ability. Here, mathematics seems to be used as a tool of persuasion, asserting that it's not ideological, but pure math. However, as we have shown, they are not proving what they claim to be proving.

Finally, note that in Proposition \ref{prop:other results} we explore the full implications of their hypotheses. If the ``random" group is large enough, it always reaches the correct solution. Recall that in any case, the group size must be large, as acknowledged by Hong and Page. Moreover, for any group size, there is always unanimity. If these results are not the case, the hypotheses of the theorem fail and it cannot be applied, which restricts its real world implications and its status as a mathematical truth.

\section{Another misuse of mathematics: Landemore's invalid and unsound epistemic argument}\label{sec:7}
\subsection{The argument}\label{sec:Landemores arg}
The argument, in a nutshell, is, \cite{Lan21}:
\begin{quote}
	Democracy is here modeled as a collective decision-procedure involving the combination of two mechanisms: inclusive and egalitarian deliberation and simple majority rule. The claim is that democracy thus defined is more likely to yield better solutions and
	predictions on political questions than less inclusive and less egalitarian decision-rules because it structurally maximizes the cognitive diversity brought to bear on collective problems. Cognitive diversity—here defined as the fact that people see problems in the world and make predictions based on different models of the way the world works or should be interpreted
	—is a group property that has been shown to be a crucial factor of group performance in various contexts and indeed more important to the problem-solving abilities of a group than individual competence of the members itself (Page 2007). We argue that under the conditions of uncertainty that characterize politics (the fact that the bundle of issues to be faced by any polity over the medium to long term cannot be predicted ahead of time), political decision-making characterized by maximal inclusiveness and equality can
	be expected to be correlated with greater cognitive diversity, which, in turn, is correlated with better problem-solving and prediction. A central assumption of the argument is that politics is characterized by uncertainty. This uncertainty (which is an assumption about the world, not necessarily the subjective epistemic stage of the deliberators) is what renders all-inclusiveness on an equal basis epistemically attractive as a model for collective decision-making. Given this uncertainty egalitarian inclusiveness is adaptive or “ecologically
	rational” (Landemore 2014).
\end{quote}
And the conclusion is:
\begin{quote}
	The argument presented here is based on a simple model of democracy and is entirely deductive. It essentially credits the epistemic superiority of democracy to inclusive deliberation, that is, deliberation involving all the members of the community (whether directly or, where 	unfeasible, through their democratic representatives) [...] The advantage of my deductive epistemic argument, ultimately, is that even if it fails to explain the way actual democracies work, it can serve as a useful normative benchmark to diagnose the way in which existing democracies epistemically dysfunction and imagine alternative institutional arrangements. One implication of the epistemic argument is indeed that in order to obtain the theoretically promised epistemic benefits of democracy, we would need to make the decision-procedures used in actual democracies a lot more inclusive and a lot more egalitarian than they are at present. Institutional reforms that the argument points toward include the replacement of elected representatives with randomly selected ones and a greater use of simple majoritarian decision-making.
\end{quote}
While the argument is not explicitly stated\footnote{Despite claiming that ``The argument presented here is based on a simple model of democracy and is entirely deductive,'' the precise premises, intermediary steps, and conclusions are never explicitly stated. This should be the first step in constructing the argument and possible replies.}, a crucial hypothesis needed for the theorem assumes the following forms:
\begin{itemize}
	\item \textit{Hypothesis}, $H$: Cognitive diversity, defined as individuals seeing problems and making predictions based on different models of the world, is a group property that improves group performance in various contexts.
	\item \textit{Hypothesis'}, $H'$: Greater cognitive diversity within a group correlates with better problem-solving and prediction abilities.
%
\end{itemize}
To justify this, Landemore relies on the results of Hong and Page as described above \cite{Lan14}:
\begin{quote}
	To make that claim, I essentially rely on Hong and Page’s formal results about the centrality of cognitive diversity to the emergent property of
	collective intelligence. 
\end{quote}
We aim to demonstrate that this hypothesis is unjustified, which subsequently renders the argument both logically unsound and inapplicable to real-world scenarios. Additionally, we will highlight instances where she incorrectly deduces propositions from these mathematical theorems, leading to a logically invalid argument.
\begin{remark}
	Some of the critiques presented in the previous section also apply to Landemore. For instance, when she informally discusses the theorem, she falls into the same misrepresentation as Hong and Page, as discussed in Section \ref{sec:answer}. She stated in a public \href{https://youtu.be/HERmRx9wDXc?t=1654}{debate}:
	\begin{quote}
		There are multiple Hong-Page theorems. The one that I use mostly is the ``Diversity Trumps Ability'' theorem. It's basically a formalization of the idea that under certain conditions, you're better off solving problems with a group of average people who think differently than with a group of experts or very smart people.
	\end{quote}
	As we have previously shown in Section \ref{sec:answer}, this is incorrect and misuses the theorem out of its range of applicability. 
\end{remark}
\subsection{Invalid argument: Misstatement of the mathematical theorems} \label{sec:Land errors}
Landemore says (about the Theorem \ref{th:HP PNAS}):
\begin{emphasisquote}
	Let me pause here to emphasize what a remarkably \textbf{counterintuitive}, indeed amazing, result this is. Where the conditions apply, you are better off with a \textbf{random group of people who think differently than with a bunch of Einsteins}! Who would have thought? In my view, this result should truly change our perspective on what makes groups smart to begin with; I believe it has huge implications for the way we should think about political bodies making
	decisions on our behalf.
\end{emphasisquote}
Also \cite{Lan14},
\begin{quote}
	That theorem was sufficiently counterintuitive that they provided a computational example to provide intuition.
\end{quote}
This misstatement is a confusion of the conclusions of the theorem with its hypotheses. The fact that a ``bunch of Einsteins'' is equivalent to only one Einstein (who, by hypothesis, cannot always solve the problem) is not a conclusion; it's an unmentioned  assumption.  More precisely, the hypotheses stipulate that the ``random" or ``diverse" group always reaches the optimal solution, see Corollary \ref{cor:diverse group}. Moreover, by assumption, a group of Einsteins is considered equivalent to one Einstein (Section \ref{sec:clone performance}). Yet again by assumption, it's not always guaranteed that this group or an individual Einstein reaches the optimal solution (Assumption \ref{as:difficult}). So it is not counterintuitive or surprising, it is merely a reiteration of the assumptions, never fully disclosed. Indeed, in, for instance, \cite{Lan12} and \cite{BL21}, it is not mentioned that clones working together are \textit{presumed} to perform just like a single person working alone (refer to Section \ref{sec:clone performance}), or that the best agent ("Einstein") is postulated to be unique (see Section \ref{sec:unique best}). Further details will be elaborated in the next section. Furthermore, by Proposition \ref{prop:other results}, the random group is likely to include a collection of Einsteins. Thus, the basic structure of the argument is, Section \ref{sec:obscuring}:

\begin{enumerate}
	\item \textit{Hypothesis 1}, $H_1$: Group $G_R$ always reach the optimal solution. $G_R$ can include a collection of Einsteins.
	\item \textit{Hypothesis 2}, $H_2$: A collection of Einsteins is not perfect.
	\item \textit{Conclusion}, $C$: Group $G_R$ is ``better" than a collection of Einsteins.
\end{enumerate}

Thus, the ``under the right conditions" of Landemore is, basically, presupposing the conclusion. Once the probabilistic component, which might be obscuring for non-mathematicians but standard for most mathematicians, is removed, the theorem is a simple reiteration of the hypothesis (see Section \ref{sec:obscuring}).
\begin{remark}
	This misstatement appears to have significant implications on Landemore's thought (from the same debate as before):
	\begin{quote}
The theorem's conclusions are not intuitive at all. I think they run against an entrenched belief that experts know best [...]. What this theorem unveiled for me is the possibility that when it comes to collective intelligence, we should stop thinking of it in terms of an addition of individual intelligences. It's really more about the group property. Does it contain enough diversity that we're going to push each other closer again to this global optimum? And that, I think, is not trivial at all. For me, it was a paradigm shift.
	\end{quote}
	Moreover, it leads her to compare the Hong-Page theorem, which is trivial (Section \ref{sec:obscuring}), with a genuinely profound and counterintuitive theorem, such as Arrow's impossibility theorem. However, she believes the difference in treatment between the theorems is based on the difference in their conclusions:
	\begin{quote}
For me, these results are remarkable. In fact, it's interesting to see that other theorems, like the Arrow's Impossibility Theorem, which leads to very negative conclusions about democracy, are considered brilliant and worth a Nobel Prize. It always seems that things are not considered surprising and trivial if they go in one particular direction.
	\end{quote}
\end{remark}

She states (see the full quote below):
\begin{quote}
This assumption ensures that the randomly picked collection of problem solvers in the larger pool is diverse—­and in particular, more cognitively diverse than a collection of the best of the larger pool, which would not necessarily be the case for too small a pool relative to the size of the subset of randomly chosen problem solvers or for too small a subset of problem solvers in absolute terms.
\end{quote}
She seems to allude to the necessity of both \(N\) (size of the larger pool) and \(N_1\) (size of the subset of ``random'' problem solvers) being sufficiently large for Theorem \ref{th:HP PNAS} to hold. However, it appears she might be misinterpreting the roles of \(N\) and \(N_1\) in the proof. While a large \(N_1\) is essential to achieve ``diversity'' (as detailed in Section \ref{sec:clone selection} or Appendix \ref{ap:simpler proof}), increasing \(N\) is not for fostering diversity. Instead, it is to ensure that the best problem solvers are identical copies of the same agent (as discussed in Section \ref{sec:clone selection}, see also Section \ref{sec:rem book hypotheses}).

Despite this, Landemore, after stating the Theorem \ref{th:HP PNAS}, says \cite{Lan21}
\begin{emphasisquote}
	To the extent that \textbf{cognitive diversity} is a key ingredient of collective intelligence, and specifically one that \textbf{matters more than average individual ability}, the more inclusive the deliberation process is, the smarter the solutions resulting from it should be, overall.
\end{emphasisquote}
As we saw in Section \ref{sec:ability trumps div}, this cannot be correct. The theorem presupposes ability in the sense that every problem solver in every state improves the state to a new state closer to the global optimum. Furthermore, as shown by the more realistic Theorem \ref{th:ability main}, if we create two groups from the group in the original theorem – one in which we make the minimal reduction in ability while ensuring full diversity, and another in which we significantly reduce diversity while ensuring ability – the less diverse group would systematically outperform the fully diverse group. In other words, ability trumps diversity.
 
There are also other severe mathematical errors with the ``Diversity Prediction Theorem", see Appendix \ref{sec:5}. Landemore says \cite{Lan21}:
\begin{quote}
	In other words, when it comes to predicting outcomes, cognitive differences among voters matter just as much as individual ability. Increasing prediction diversity by one unit results in
	the same reduction in collective error as does increasing average ability by one unit.
\end{quote}
This is mathematically incorrect and entirely wrong: the effect is undetermined, it's not of the same magnitude, and it's not necessarily a reduction, as explained in Section \ref{sec:asymmetric}, see Error \ref{obs:trivial obs}, and Section \ref{sec:Pages error}. It is a significant mathematical error to assume that the terms in Theorem \ref{th:div pred} can be changed while the other remain fixed. Furthermore, as we observe in Proposition \ref{prop:SEmax}, the diversity and ability terms do not play the same role. While increasing ability eventually reduces the prediction error, increasing diversity does not have the same effect and, furthermore, it eventually increases the maximum prediction error. Therefore, Landemore's argument has a significant gap; without controlling for ability, increasing diversity does not guarantee a reduction in the prediction error.
\subsection{Invalid argument: Page and Landemore's misstatement of the theorem's hypotheses} \label{sec:rem book hypotheses}
To justify the use of the theorem, Landemore states \cite{Lan12},
\begin{emphasisquote}
	Importantly, the four conditions for this theorem to apply \textbf{are not utterly demanding}. The first one simply requires that the problem be difficult enough, since we do not need a group to solve easy problems. The second condition requires that all problem solvers are relatively smart (or not too dumb). In other words, the members of the group must have local optima that are not too low; otherwise the group would get stuck far from the global optimum. The third condition simply assumes a diversity of local optima such that the intersection of the problem solvers’ local optima contains only the global optimum. In other words, the participants think very differently, even though the best solution must be obvious to all of them when they are made to think of it. Finally, the fourth condition requires that the initial population from which the problem solvers are picked must be large and the collection of problem solvers working together must contain more than a handful of problem solvers. This assumption ensures that the randomly picked collection of problem solvers in the larger pool is diverse—­and in particular, more cognitively diverse than a collection of the best of the larger pool, which would not necessarily be the case for too small a pool relative to the size of the subset of randomly chosen problem solvers or for too small a subset of problem solvers in absolute terms.
\end{emphasisquote}
	However, we need more hypotheses than the one mentioned by Landemore. As Landemore acknowledges, she is following Page's book. But he does not mention important hypotheses of his result. In particular, Page fails to mention Assumption
	\ref{as:injective}, \ref{as:ability},
	\ref{as:idempot},
	\ref{as:unique best agent},
	\ref{as:group del},
	\ref{as:clones}, and \ref{as:random}. As we saw in Sections \ref{sec:V inj}, \ref{sec:unique best}, \ref{sec:clone performance}, and \ref{sec:clone selection}, if these conditions do not hold, the theorem doesn't hold (see the counterexamples). And, as we've seen, these conditions can be rather restrictive (such as assuming that a billion Einsteins would not outperform a single Einstein). However, in his book or in a shorter self-contained article, see \cite{Pag07}, his exposition only mentions four conditions:
	\begin{itemize}
		\item \textbf{Condition 1 : The Problem Is Difficult:} No individual problem solver always locates the best solution.
		\item \textbf{Condition 2 : The Calculus Condition:} The local optima of every problem solver can be written down in a
		list. In other words, all problem solvers are smart.
		\item \textbf{Condition 3: The Diversity Condition:} Any solution other than the global optimum is not a local optimum for
		some non-zero percentage of problem solvers.
		\item \textbf{Condition 4: Reasonably Sized Teams Drawn from Lots of Potential Problem Solvers:} The initial
		population of problem solvers must be large, and the teams of problem solvers working together must consist
		of more than a handful of problem solvers.
	\end{itemize}
	And Page concludes as \cite{Pag07}:
	\begin{quote}
		\textbf{The Diversity Trumps Ability Theorem:} Given conditions l to 4, a randomly selected collection of problem solvers almost always outperforms a collection of the best individual problem solvers.
		
		This theorem is no mere metaphor, cute empirical anecdote, or small-sample empirical effect that may or may not be true with more trials. It's a logical truth like the Pythagorean Theorem
		(Hong \& Page, 2004) [i.e., \cite{HP04}]. 
	\end{quote}
	
	Moreover, his Condition 2 (Landemore's second condition; see also \cite{Lan14}) is ill-stated. The ``Calculus Condition'' requires that $\phi(X)$ is countable (which is trivial if $X$ is finite), but he interprets it as ``all problem solvers are smart.'' This condition doesn't relate to being smart, contrary to Page's and, consequently, Landemore's interpretation. For instance, consider the function $\phi:X\to X$ defined as $\phi(x)=x_m$, where $V(x_m)$ is a global minimum of $V$ (e.g., 0). Then, $\phi(X)={x_m}$, finite, which hardly represents being ``smart.'' In fact, it's the worst agent conceivable, the less smart, since it assigns the solution furthest from the global optimum to every state. Nevertheless, Page (and subsequently Landemore) refers to this as being ``smart.''
	
	It's also noteworthy that in his book, Page's conditions, such as Condition 3 above, are subject to Thompson's critique (see Section \ref{sec:V inj}). Indeed, in his response to Thompson, Page does not refer to his stated Condition 3, but to what he incorrectly thinks Condition 3 requires (which is the same mistake present in \cite{HP04} and pointed out by Thompson). He says in \cite{TP15}:
	\begin{quote}
		The article “Does diversity trump ability?” (Notices, October 2014) characterized that claim as false. That characterization was based on an erroneous counterexample that violates my theorem’s Condition 3, (specified in my book): for any nonglobal optimum, some positive proportion of the problem solvers can locate a solution of higher value.
	\end{quote}
	But that quote does not appear in his book, and when he defines Condition 3 in his book, he says what we quoted above, which is subject to Thompson's critique. Furthermore, as we have seen, relying on the book for a fair statement of the theorem hypothesis is far from appropriate as many of them are missing.
	
	Finally, similarly as Landemore, his justification of Condition 4 misrepresents the nature of the theorem. He states:
	\begin{quote}
		If the initial set consists of only 15 problem solvers, then the best ten should outperform a random ten. With so few problem solvers, the best ten cannot help but be diverse and therefore have different local optima. At the same time, the teams that work together must be large enough that the random collection can be sufficiently diverse. Think of it this way: We need to be selecting people from a big pool [$N$ large], and we need to be constructing teams that are big enough [$N_1$ large] for diversity to come into play.
	\end{quote} 
	As we saw in Section \ref{sec:clone selection}, the agents, an $N_1$ number of them, are not taken from the same larger pool of $N$ elements, but from two different independent events. In fact, there is no larger pool with a finite size, but we take elements from $\Phi$, and if an element already appeared, it can appear again, i.e., there is a possible infinite amount of clones, Assumption \ref{as:clones}. 
	
	It is not a proper pool (with a finite size) because we can take an infinite number of elements from it. Let us call it the infinite pool of clones. As detailed in that section, the actual process is the following. First, we form the ``random" group taking elements from this infinite pool of clones until we get the desired elements (enough to outperform the best agent). Note that the number of elements of this group, $N_1$, is not fixed; it is not always 10 or 15, but it can be as large as needed depending on the situation. That is, this number is random and defined in such a way that cancels the randomness previously introduced: it is defined so that we always have a ``random group" containing the elements we want. Once this group is formed, then, from an independent event, we make a large set of clones, of size $N$ such that there are $N_1$ clones of the best agent. Again, $N$ is not a fixed number, but as large as needed until the desired conclusion is reached. This has nothing to do with choosing from a given pool with finite elements, as Landemore and Page seem to suggest. What is more, it is not just that the description is not accurate, but as we showed in Section \ref{sec:clone selection}, doing something similar to what they seem to suggest would make the theorem false. We cannot select the agents in a normal way, but we need an infinite amount of clones of each agent to be able to stop whenever we have reached the desired elements. That is the reason why we argued that they introduced the probability apparatus unnecessarily, Section \ref{sec:obscuring}, which can be simplified by removing the probability part. The simplification reveals that the theorem is merely a restatement of the hypotheses, see also Section \ref{sec:simpler proof}.
\subsection{Unsound argument: The misuse of hypotheses of the ``Diversity Trumps Ability Theorem"}\label{sec:Land hyp}
As we have seen in the previous section, with respect to the four conditions stated by Landemore and Page, those are not the only conditions required for the theorem to apply. Among others, she doesn't mention the hypotheses from Sections \ref{sec:V inj}, \ref{sec:unique best}, \ref{sec:clone performance}, and \ref{sec:clone selection}. If these conditions do not hold, the theorem doesn't hold (see the counterexamples). And, as we've seen, these conditions can be rather restrictive, such as assuming that a billion Einsteins will not outperform a single Einstein. Therefore, her statement of the theorem is incorrect, and the hypotheses of the theorem are more restrictive, which makes her subsequent analysis incorrect. Furthermore, as discussed in Section \ref{sec:Land errors}, Landemore incorrectly deduces some conclusions from the theorems that, although favorable to her position, are mathematically mistaken.

For the sake of dialectical purposes, we are going to omit those errors (although they are enough to invalidate the argument) and add another layer to the critical argument. Even if the theorem were not ill-stated, we have other problems: the hypotheses are not plausible, and they lead to a logical contradiction. Indeed, Landemore claims that the conditions of the theorem (Theorem \ref{th:HP PNAS}) are generally satisfied, as seen in the previous quote or the section ``The meaning and empirical plausibility of the assumptions behind the Diversity Trumps Ability Theorem" in \cite{Lan14}. However, other sets of hypotheses are plausible, or even more so, which is problematic. More specifically, in a Moorean style, we are going to construct a set of incompatible propositions, requiring us to reject (at least) the least plausible one. We are going to use Theorem \ref{th:ability main} for this task.

The propositions are as follows:
\begin{enumerate}
	\item Hong-Page's framework can be used for a deductive argument for epistemic collective-decision systems in the sense that it can serve as a benchmark or be useful in deriving implications to obtain epistemic benefits (as in Section \ref{sec:Landemores arg}).
	\item The assumptions of Theorem \ref{th:basic HP}.
	\item The assumptions of Theorem \ref{th:ability main}.
\end{enumerate}
Note that (1) and (2) imply that ``diversity trumps ability," but (1) and (3) imply ``ability trumps diversity", so at least one of these propositions must be rejected, but:
\begin{enumerate}
	\item Rejecting the first would invalidate Landemore's argument, as the theorem would then have no relevance to collective-decision systems.
	\item Rejecting the second would undermine Landemore's proposition that ``cognitive diversity is a key ingredient of collective intelligence, and specifically one that matters more than average individual ability".
	\item Rejecting the third, without rejecting (2), amounts to ``biting the bullet". The assumptions of Theorem \ref{th:ability main} are relatively more plausible than those in (2). For instance, there is no need to assume the existence of clones, that 100 Einsteins working together to solve a problem are the same as one, or that there will be no disagreement. Furthermore, unlike Hong and Page's Theorem, it provides a fair comparison between ability and diversity. See Section \ref{sec:obscuring} and references therein for more details.
\end{enumerate}
Note that we are not claiming that Theorem \ref{th:ability main} has any social content (we personally reject the first proposition), but it suffices to form the Moorean set of incompatible propositions. It also serves to show how Landemore commits an equivocation fallacy or misuses natural language to represent mathematical statements. For instance, Assumption \ref{as:ability} is read as agents are ``relatively smart (or not too dumb)" and then she discusses that voters satisfy this, \cite{Lan14}. But note that the hypotheses of Section \ref{sec:div group} only reduce by an almost insignificant amount the ability, so they can be considered ``relatively smart (or not too dumb)". But the thesis changes radically. Thus, Landemore's justification for the use of the theorem is seriously compromised.

Finally, even assuming that the hypotheses are plausible, there exists a significant contradiction in Landemore's work. Recall that the hypothesis of Theorem \ref{th:HP PNAS} not only guarantees that the ``random" group is better than the best agent, but also ensures that they always reach the correct conclusion without disagreement or dissent, as shown in Proposition \ref{prop:other results}, see also Remark \ref{rem:unanimity}. These are the same hypotheses that, according to Landemore, make cognitive diversity more crucial than individual ability. This perfect deliberation is what enables the ``diverse" group to surpass the best agent; the former is perfect by assumption, while the latter is not. Nonetheless, Landemore states in \cite{Lan21}:

\begin{quote}
	Deliberation is far from being a perfect or complete decision-mechanism, in part because it is time-consuming and rarely produces unanimity.
\end{quote}

And in \cite{Lan14}, she further notes:

\begin{quote}
	I thus do not need to assume away, as Quirk seems to accuse me of doing, the possibility of disagreement.
\end{quote}

Therefore, if she rejects an implication of the theorem, she must also reject at least one of its hypotheses. However, as we have seen, she defends the hypotheses of the theorem she applies. This creates a contradiction that is never addressed.

\subsection{The substantive content of the Numbers Trump Ability ``Theorem"}\label{sec:vacuousness}
Landemore's key innovation is the following, as stated in \cite{Lan14}:
\begin{quote}
	The second step of my argument—my addendum to Page and Hong—
	proposes that the “cheapest” (i.e., easiest and most economical) way to
	achieve cognitive diversity in the absence of knowledge about the nature of
	complex and ever-changing political problems is to include everyone in
	the group. [...] This “Numbers
	Trump Ability Theorem” thus supports a strong epistemic case for
	democracy, in which my key innovation is to support inclusiveness for its
	instrumental, specifically epistemic properties: Under the right conditions,
	including everyone in the decision-making process simply makes the group
	more likely to get the right (or, at least better) answers.
\end{quote}
Her argument is straightforward: if, for epistemic reasons, diversity is what matters, then including everyone is the simplest way to increase diversity. Aside from practical issues, which Landemore somewhat considers, the problem with this reasoning (which is not actually a theorem) is that the premise is false. We have argued that in both Hong-Page theorems, ability plays a crucial role, as seen in Section \ref{sec:ability trumps div} and \ref{sec:asymmetric}. Thus, increasing the number of people can have detrimental effects. Therefore, the ``theorem" is false. However, by maintaining the Hong-Page framework and considering the observations above, it can be ``corrected" as:

\begin{theorem}[Enlightened Numbers Trump Numbers ``Theorem"]
	Under the right conditions and given the uncertainty in the ability of the agents, including everyone with ability above a certain threshold in the decision-making process makes the group more likely to arrive at the correct (or, at least, better) answers than merely including people without controlling for ability.
\end{theorem}
If this uncertainty is one of the reasons used to prove Landemore's ``theorem", it can backfire and serve, while maintaining the same framework as Landemore's, as a reason to reach an opposite conclusion to inclusiveness: If we acknowledge the ``absence of knowledge about the nature of complex and ever-changing political problems'', it would be prudent to select problem solvers who are competent enough to handle these uncertain problems. Hence, we must take ability into account. \textit{In other words, once corrected, this theorem lends support to a version of epistocracy}. As we have previously stated, we do not find the Hong-Page theorems particularly enlightening, so we do not advocate for this theorem. Nonetheless, if we follow Landemore's line of reasoning, this interpretation would be more accurate.

In general, all the theorems that Landemore uses for her political defense of democracy (\cite{Lan12}) presuppose certain levels of ability and knowledge. This is the case of the Hong-Page theorems, as seen in Section \ref{sec:ability trumps div} and \ref{sec:asymmetric}, as well as the Condorcet Jury Theorem (CJT) and the Miracle of Aggregation, as shown in \cite[Theorem 3.5 and Theorem 4.3]{Rom22}. These latter two theorems, which belong to the same general theorem (non-homogeneous CJT), are far more likely if we include epistemic weights that are stochastically correlated (with a measurement error) with epistemic rationality. Also, if ability is not controlled, these theorems can operate in the opposite direction, ensuring that we almost surely choose the \textit{wrong} option. Thus, from an epistemic and instrumental perspective, these theorems strongly suggest including ability thresholds or, the more feasible and semiotically problem-free case of epistemic weights with a minimum of 1 (no one excluded) and stochastically correlated (the inevitable measurement error is taken into account, perfect correlation is not assumed) with epistemic rationality, for a starting practical proposal, see \cite[Appendix D]{Rom22} or, for a lengthier discussion\footnote{In Spanish, but see references therein (in English).}, \cite{Rom21}. This could serve as a preliminary proposal that needs to be tested and experimented with. While it might still be far from perfect, it should be evaluated in comparison to the existing alternatives.

Nevertheless, Landemore strongly opposes epistocracies. In her chapter ``Against Epistocracies", \cite{BL21}, she says:

\begin{quote}
	My first question to Brennan is this: What would such
	exclusion achieve? Recall that in my model deliberation
	does most of the epistemic work. Most filtering of bad
	input or bad reasoning occurs at that deliberative stage.
	So there is no reason not to include everyone as one more,
	howsoever uninformed, voice will not pollute the outcome but will at most delay the conclusion of the deliberation.
\end{quote}

This is incorrect. First, if no selection is done, one cannot ensure the conditions of the Hong-Page theorem, so one cannot expect the result (that deliberation works) to hold. Second, as we have seen in Theorem \ref{th:ability main}, introducing these kinds of agents can pollute the outcome, getting stuck at a solution far from the global maximum. This is the same error that pollutes all of Landemore's analysis, so we insist here; all these theorems assume a certain amount of ability, but Landemore just presupposes\footnote{For instance, \begin{quote}
		Assuming that, on average, the citizens from among
		whom we select representatives meet a minimal thresh-
		old of individual competence, random selection is a more promising, authentically democratic way of selecting representatives that maximizes cognitive diversity in the 
		face of political uncertainty.
\end{quote} See Section \ref{sec:Land hyp} for more.} this without questioning seriously enough and  focusing mainly on diversity, which has a ``secondary" effect (Theorem \ref{th:ability main} and Section \ref{sec:asymmetric}).

For instance, she continually emphasizes the uncertainty:

\begin{quote}
	As time goes by and circumstances change, however, it becomes very likely that his epistocracy will run into issues where it will miss the very voices and votes it purposely excluded. Even if the probability is low, the expected cost might still be huge. Why
	take the risk?
	There may be a short window of time in which a
	Brennanist epistocracy would work, perhaps even better
	than a democracy. But probabilistically, this superiority
	is bound to vanish over time. The question is when.
	[...]
	Most importantly, there is no reason
	to exclude any voice in a model that assumes democratic
	deliberation itself can weed out the bad input.
\end{quote}

However, this same uncertainty, when translated into uncertain abilities of the problem solvers, could lead to the inclusion of some problem solvers who, rather than aiding, actually obstruct us from reaching the optimal solution. Thus, her claims like ``But probabilistically, this superiority is bound to vanish over time" and that the expected cost is substantial are unfounded, and she provides no valid proof for such strong propositions. It's important to note that there may be merit in including all voices in some capacity. The purpose of this part is not to criticize that, but to critically analyze her use of mathematical results to draw certain conclusions. Therefore, it is not reasonable—and borders on begging the question—to assume that democratic deliberation itself can weed out bad input without any further justification for this strong result.
\section{Conclusion} \label{sec:8}
Our thorough analysis of the Hong-Page Theorems has revealed several critical issues. These include the portrayal of the theorem as an unnecessary mathematical formalization of a simple fact, effectively a mere reiteration of the hypothesis, overlooking certain assumptions, and the application of mathematical principles in ways that may not fully align with their desirable use in sociopolitical contexts. We have concluded that what the theorem requires is not ``diversity", but the existence of a more ``able" problem solver who can improve upon areas where others fall short.

The application of mathematics by Hong and Page in their theorem presents its findings in a manner that may hide the straightforward nature of its conclusions. Through the use of mathematical formalism, aspects that might be considered basic are portrayed with a level of complexity that suggests deeper insights, potentially leading to a misunderstanding of the theorem's implications. It is crucial to approach the application of mathematics with careful consideration and precision, particularly when it informs the basis for decisions with significant societal impact. Despite its wide recognition and numerous citations, the Hong-Page theorem, as well as the related ``Diversity Prediction Theorem," should be critically evaluated for their contributions to the understanding of collective decision-making. This careful scrutiny is also applicable to Page's book \cite{Pag08}, where the foundational claims merit a reexamination in light of these considerations. From\footnote{Quotes are included here to demonstrate how the main contributions, according to the authors (hence the use of quotes), of their respective books are significantly impacted by the analysis presented in this paper, essentially indicating that their validity is considerably compromised.} the preface:

\begin{quote}
Perhaps because \textit{The Difference} takes time to digest, eventually, accurate readings won out. Reviewers recognized that\textit{ The Difference} explores the pragmatic, bottom-line contributions of diversity. It does so using models and logic, not metaphor. The book’s claims that “collective ability equals individual ability plus
diversity” and that “diversity trumps ability” are mathematical
truths, not feel-good mantras.
\end{quote}

Hélène Landemore's utilization of mathematical reasoning within her political proposal and epistemic argument encounters similar challenges regarding logical validity and soundness. Her ``Numbers Trump Ability" theorem, inspired by the Hong-Page Theorems, along with other derived conclusions such as her interpretation of the ``Diversity Prediction Theorem", face critical scrutiny. Our analysis suggests that her epistemic argument encounters foundational issues. Thus, the core argument of her book \cite{Lan12} faces significant challenges. As stated by Landemore in \cite{Lan14}:

\begin{quote}
	Let me briefly rehearse what I see as the
	main argument of the book. At its heart is a simple model of what, under
	certain conditions that I deem plausible enough, can be expected of an
	inclusive political decision process in a comparison with less inclusive
	ones. [...] In my eyes, the main value of my book is to create a simplified,
	relatively rigorous framework for the meaningful comparison of the
	properties of basic political “regimes.”
\end{quote}
A similar problem arises in other foundational claims related to her political proposal of ``Open Democracy'', found in works such as \cite{Lan12}, \cite{Lan14}, \cite{Lan21}, and \cite{BL21}.

This critique should not be taken as a dismissal of the importance of diversity in decision-making, but rather as a call to address the misuse of mathematics in these contexts. It urges us to consider the rigorous and nuanced approach required when applying mathematical theories to sociopolitical constructs. As such, this paper makes a contribution to the ongoing discourse on collective intelligence, fostering a deeper understanding of uses and misuses of mathematical formalization in social theory.

\newpage
\appendix
\section{Simpler proof of Theorem \textnormal{\ref{th:HP PNAS}} and discussion of the role of $N_1$} \label{ap:simpler proof}
We collect some technical issues related to Theorem \ref{th:HP PNAS} in this appendix.
\subsection{Simpler proof}
\begin{proof}[Proof of Theorem \textnormal{\ref{th:HP PNAS}}]
	First, by hypothesis (Assumptions \ref{as:difficult} and \ref{as:diversity}), $\exists~x_*\in X,~\phi^*,\phi_*\in\Phi$ such that the best agent $\phi^*(x_*)\neq x^*$ and $V\left(\phi_*\left(\phi^*(x_*)\right)\right)>V\left(\phi^*(x_*)\right)$. By hypothesis (Assumptions \ref{as:group del} and \ref{as:clones}), $V\circ\phi^{\{\phi^*,~\phi_*\}}\ge V\circ\phi^*$, where the equality is strict for at least one point. Given that $\nu$ has full-support, $\mathbb{E}_\nu\left( V\circ\phi^{\{\phi^*,~\phi_*\}}\right)> \mathbb{E}_\nu\left(V\circ\phi^*\right).$
	
	Second, we introduce the probabilistic selection of clones. By the Strong Law of Large Numbers (SLLN),
	$$
	\mu\left(\omega\in\Omega~:~\bigcap_{\phi}\left(\lim_{N\to\infty}f^N(\phi)=\mu\left(\{\phi\}\right)\right)\right)=1\,,
	$$
	where $f^N\left(\phi\right)$ represents the frequency of appearance of $\phi$ when the size of the group of clones is $N$. The intersection is finite. For this full-measure set, we define $N_\phi=N_\phi(\omega)$ as the integer such that, if  $N\ge N_\phi$, then $f^N\left(\phi\right)>\mu\left(\{\phi\}\right)/2.$ Following Assumption \ref{as:random}, we take $N_1\coloneqq\max\{N_{\phi^*},N_{\phi_*}, \frac2{\mu\left(\{\phi^*\}\right)}, \frac2{\mu\left(\{\phi_*\}\right)}\}$ for the first event. By these definitions, at least one copy of $\phi_*$ and $\phi^*$ are included. For the second event, take $N\ge \frac2{\mu\left(\{\phi^*\}\right)}N_1$, so there are more than $N_1$ copies of $\phi^*$ in the second group. The proof then follows from the first part of this argument.
\end{proof}
See the main text, Section \ref{sec:simpler proof}, for an intuitive explanation of this proof.
\subsection{The role of $N_1$}\label{rem:N1 large}
	Noting that if $X$ is finite, then $\Phi$ must also be finite, we can choose $N_1$ such that almost surely every member appears in the random group. We just need to set:
\begin{equation}\label{eq:def gen N1}
		N_1\coloneqq\max_{\phi\in\Phi}\lbrace N_{\phi}, \frac2{\mu\left(\{\phi\}\right)}\rbrace\,.
\end{equation}
	In the original proof by Hong and Page, $N_1$ is set so that every member needed to outperform the best performing agent appears with probability one, cf. Remark \ref{rem:Phi0}. Thus, $N_1$ must be large, which virtually\footnote{In the extreme case where $\phi_0$ is not needed at all and $\mu\left(\lbrace\phi_0\rbrace\right)$ is close to zero, we cannot ensure that, even for a large $N_1$, one copy is included almost surely.} guarantees that at least one copy of each member of $\Phi$ is included in the ``random group''. In any case, this $N_1$ as defined is sufficient to ensure the theorem holds true.

	More precisely, we can consider three different $N_1$ so that the following groups are included almost surely:
	\begin{enumerate}
		\item $\Phi$, all agents from the given pool or $\Phi_0\subset \Phi$, a subset of agents that is always sufficient to attain the solution $x^*$,
		\item $\tilde{\Phi}_0\subset \Phi_0$, a group that outperforms the best performing agent, i.e, for a small $\varepsilon_1$, 
		$
		\mathbb{E}\left(V; \tilde{\Phi}_0\right)>1-\varepsilon_1\ge\mathbb{E}\left(V; \{\phi^*\}\right)\,,$
		\item $\{\phi^*,\phi_*\}\subset\Phi$, a pair of agents that is enough to outperform the best-performing agent, $\phi^*$.
	\end{enumerate}
	
	For each subset, you can define, using \eqref{eq:def gen N1}, a $\{N_1^{(i)}\}_{i=1,2,3}$. Each one is sufficient for the theorem to hold.	But, let us consider the following:
	
	\begin{itemize}
		\item If our goal is merely to outperform the best agent, then  $N_1^{(3)}$ from the simpler proof above should suffice and is more straightforward.
		\item Given the same hypotheses, one could simply select $N_1^{(1)}$ and assert a stronger result, as it has been demonstrated in Proposition \ref{prop:other results}. It is not just that they outperform the best agent, but they also almost surely find the optimal solution. So, with the same hypotheses, one can prove a stronger result. Shouldn't one then fully explore the implications of these hypotheses?
	\end{itemize}
	Given that $N_1=N_1(\omega)$ is \textit{not fixed} but rather a ``stopping time" that halts when the desired members, as per the selector's or chooser's preference, are included, one might ask why the selector should stop merely to outperform the best agent if the correct solution can always be reached when the hypotheses are met. This more comprehensive result trivially implies the one in  \cite{HP04}: if there exists an $N_1$ such that the correct solution is always reached, then, by monotonicity, there must exist an $N_1$ such that the other group is outperformed.
	
	Also worth noting is that one cannot exclude the possibility of other members, different from the ones in the selected group, appearing. For instance, although we may choose $\tilde{\Phi}_0$ such that $\phi^*\notin \tilde{\Phi}_0$, the ``random" group might eventually include $\phi^*$ because of the way these members are chosen, as detailed in Section \ref{sec:clone selection}.
	This is why we consistently refer to $\Phi_0$ or $\Phi$, as these capture the full consequences of the hypotheses. In the face of uncertainty, the only certain way to include all necessary agents is to include everyone. 
	
	In any case, the triviality of the theorem is manifest for any $N_1$. Indeed, its basic structure is as follows:
	
\vspace{1.7mm}
\noindent\textsc{Theorem}. \textit{{Hypotheses}}: Let there be a pool of agents such that:
	\begin{itemize}
		\item they never degrade the solution and where there is always an agent who can improve the solution,
		\item the best performing agent is imperfect by hypothesis,
		\item as it follows from these two points above, we can choose agents randomly until they outperform the imperfect best performing agent.
	\end{itemize}
	{\textit{Thesis}}: Consequently, this ``random" group will outperform the best performing agent.
	
\subsubsection{Numerical example} 
For instance, taking the numbers from an example based on Page's slides, see Section \ref{sec:answer}, with 12 agents, in the vast majority of cases, the best agent will be included in the random group using $N_1^{(2)}$. Additionally, the average $N$ will have the same order of magnitude in all three cases—$N_1^{(1)}$ is not huge—with the following inequalities:
$$
17\approx\mathbb{E}\left(N_1^{(3)}\right)<24\approx\mathbb{E}\left({N}_1^{(2)}\right)<36\approx\mathbb{E}\left(N_1^{(1)}\right)\,,
$$
the middle group being sufficient to outperform the best performing agent and consisting in six agents, as shown in Page's slide. For example, if we take $N_1^{(3)}$, the smallest one on average, a more equitable comparison would be as follows, see Figure \ref{fig:pageecb_comparison}. In the same figure we show what happens if one uses ${N}_1^{(2)}$, another group that outperforms the best agent. Also note that in that case, the probability of members other than the ones of this group not appearing is practically zero, so they must be represented in the figure in some way.

\begin{figure}
	\centering
	\begin{subfigure}{0.49\linewidth}
		\centering
		\includegraphics[width=\linewidth]{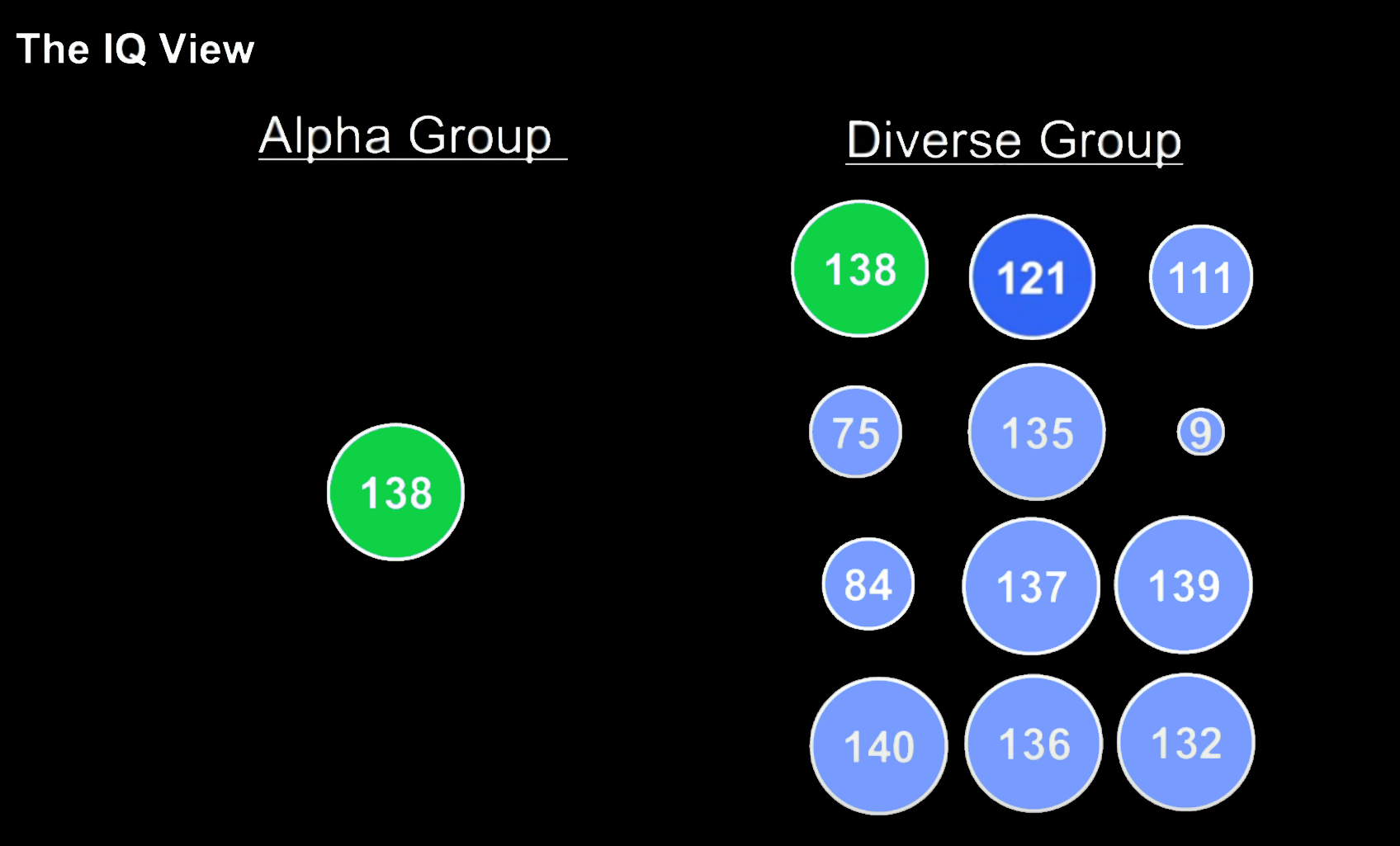}
		\caption{A presentation slide that better encapsulates the essence of the theorem, $N_1=N_1^{(3)}$.}
		\label{fig:pageecb_N2}
	\end{subfigure}
	\hfill
	\begin{subfigure}{0.49\linewidth}
		\centering
		\includegraphics[width=\linewidth]{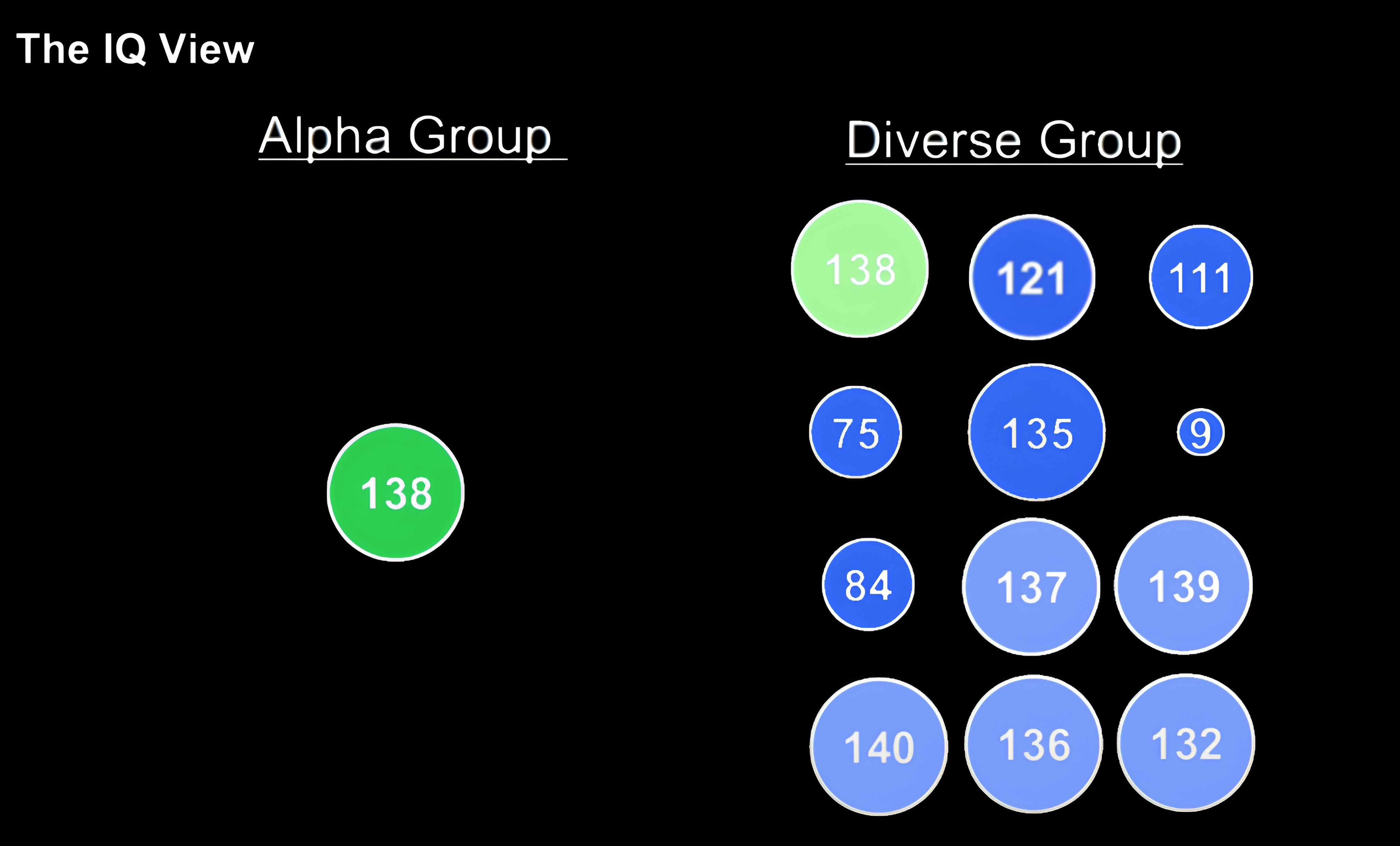}
		\caption{A presentation slide that better encapsulates the essence of the theorem, $N_1=\tilde{N}_1^{(2)}$. }
		\label{fig:pageecb_N3}
	\end{subfigure}
	\caption{Different slides. In green, the best agent; in blue, the other agents. Light colors indicate that the probability of appearing is not one (but could be close to one). In the first figure, ``121" is assumed to be complementary to the best agent; it can improve where the latter falls short. In the second figure, we are considering that the elements in dark blue are sufficient, by hypothesis, to outperform the agent ``138".}
	\label{fig:pageecb_comparison_app}
\end{figure}
	
\section{Ability trumps diversity Theorem}
\label{ap:ability trumps div th}
We give the precise statements and proof of the theorem mentioned in Section \ref{sec:ability trumps div}.
\subsection{The new assumptions}
Among a set of agents $\Phi$, we select two finite groups with different properties. We are going to modify some assumptions, but the other remain the same. First, let us introduce the possibility of disagreement following Assumption \ref{as:group del} as:
\begin{equation*}
	\phi_{i_{j+1}^x}\left(\phi_{i_{j}^x}(x_{j-1})\right)=x_{j-1}\,, ~~\text{ with }~~ \phi_{i_{j}^x}(x_{j-1})\neq x_{j-1}\text{ and }i_{j}^x\neq i_{j+1}^x\,.
\end{equation*}
A disagreement is a stopping point. In other words, if there is a cycle such that one agent proposes a new solution and other reverse back that solution there is a disagreement and that initial solution is given as the group solution. This is a simple model where disagreement is possible. Clones are not needed now, so idempotency is not necessary.
\begin{remark}\label{rem:unanimity}
	Note that in the original formulation of the Hong-Page Theorem, for any group of any size, even if they might not be able to reach the correct solution, Remark \ref{rem:Phi0}, there will be no disagreement in any case, as per Assumption \ref{as:group del}. This always leads to unanimity, which is highly unrealistic.
 \end{remark}

Let also $\mu_x$ be a probability measure such that if $x$ is the previous solution, $\mu_x\left(\{i\}\right)$ represents the probability that $\phi_i(x)$ the next solution in the deliberation chain, see Assumption \ref{as:group del}. This measure has full-support, no one is silenced. The indices are chosen independently. Once fixed the $x_0$, this defines a probability measure $\mathbb{P}$ on the possible paths. That is, 
$$
\mathbb{P}\left(x_{k}=x'\mid x_{k-1}=x\right)=\mu_{x}\left(\{i ~\mid~\phi_{i}\left(x\right)=x'\}\right)>0\,.
$$

\subsubsection{The ability group.} Let us denote a group by $\Phi^\textnormal{A}=\{\phi_\alpha\}_{\alpha\in A}$ which is selected such that:
\begin{itemize}
	\item \textit{Ability}: $V\circ\phi_\alpha\ge V$. In other words, this group is chosen to ensure ability in the sense that each agent does not decrease the value of the initial state given. This is Assumption \ref{as:ability}, but now is imposed by the selection of the group.
	\item \textit{Common knowledge}: $\exists X_{\textnormal{CK}}\subset X$ such that:
	\subitem \textit{Non-diversity set}: $\forall \alpha,\alpha'\in A$ we have $\phi_\alpha\vert_{X_{\textnormal{CK}}}\equiv \phi_{\alpha'}\vert_{X_{\textnormal{CK}}}$.
	\subitem\textit{\textit{Knowledge}}: $\phi_\alpha (x_c)=x^*~~\forall~x_c\in X_{\textnormal{CK}}$ $\forall~\alpha\in A$.	
	
	\noindent In other words, this group selection comes with a selection bias. The agents have a common knowledge that makes them similar; for the set $X_{\textnormal{CK}}$, they all give the same solution. This is an extension of the second part of Assumption \ref{as:ability}; agents are not only able to recognize that $x^*$ is the solution, but they can do the same for other states $x\in X_{\textnormal{CK}}$. Note that $x^*\in X_{\textnormal{CK}}$. 
	\item \textit{Smaller diversity set}: $\forall~x\in X\backslash X_{\textnormal{CK}}$, Assumption \ref{as:diversity} holds. In other words, for this group, the original assumption of Hong and Page only holds outside the ``larger'' set $X_{\textnormal{CK}}$, instead of $x^*$. 
	\item More importantly, we diminish diversity in a second way. For all $x$ in $X\backslash X_{\textnormal{CK}}$, there exists exactly one agent, $\phi^x$, who provides a distinct answer, and the set of unique answers could be equidistributed, meaning it's not just one agent always giving the different answer. Formally, $\{x~|~\phi^x=\phi\}\le \frac{|X|}{|\Phi^\textnormal{A}|}+1$ for all $\phi\in\Phi^\textnormal{A}$. Therefore, if the ratio is small enough, two agents are quite similar, signifying a lack of diversity, i.e., following \cite[Appendix]{HP98}, their distance is relatively small.
\end{itemize}
\begin{remark}
	For the theorem to function, we don't require a large $X_\text{CK}$, but having it large makes the agents less diverse. It could be just $x^*$ as in the original theorem.
\end{remark}
\subsubsection{The diversity group.}\label{sec:div group} Let us denote a group by $\Phi^\textnormal{D}=\{\phi_j\}_{j\in \mathcal{J}}$ which is selected such that the maximum diversity is guaranteed. More precisely, there is a unique $x^0\in X$ such that:
\begin{itemize}
	\item \textit{Full-diversity with ability}: $\forall~x\in X\backslash\{x^0, x^*\}$ there is a set of agents $\{\phi_{j^x_k}\}_{k=1}^{n_x}\subset\Phi^\textnormal{D}$ such that $\phi_{j^x_k}\left(x\right)\neq x$ and such that all the states, and only those, closer to the solution $x^*$ (that is, all improve the state) are the local optimum for some agent.
	\item \textit{Minimal ability loss}: there is only one agent $\phi_{j_0}\in\Phi^\textnormal{D}$ and only one state $x^0$ such that $V\left(\phi_{j_0}(x^0)\right)<V\left(x^0\right)$. Note that this is the \textit{minimal ability that can be lost.}
\end{itemize}

\subsection{The theorem}
\begin{theorem}[Ability Trumps Diversity]\label{th:ability} Let $\Phi^\textnormal{A}$, $\Phi^\textnormal{D}$ as above with the given assumptions. Then, the ability group outperforms the diversity group.
\end{theorem}
\begin{proof}
	To prove this theorem, we need to compare the  performances of the two groups. First, we consider the ability group $\Phi^\textnormal{A}$. Any agent from $\Phi^\textnormal{A}$ does not decrease the value of the given state. Moreover, for any state in the non-diversity set $X_{\textnormal{CK}}$, all agents in $\Phi^\textnormal{A}$ will return the optimal solution $x^*$. Thus, following the measure $\mu_{x'}$, for any $x\in X$,
	$$
	V(x)\le V\left(\phi_{i_{1}^x}(x)\right)\le\ldots \le V\left(\phi_{i_{n}^x}(x_{n-1})\right)=1\,.
	$$
	This is because $p_{x'}\coloneqq \mu_{x'}^A\left(\alpha\in A~\mid~\phi_\alpha({x'})\neq {x'}\right)>0$. This holds true even in the worst-case scenario where all agents, except one, are stuck at that point. Thus being stuck has probability, by the subadditive property,
	\begin{equation*}
		\prod_{i=1}^\infty(1-p_{x'})=0\,\rightarrow \mathbb{P}\left(\exists~n_0, x' ~:~\phi_{i^x_{n}}(x')=x'~~ \forall~n\ge n_0\right)\le \sum_{n_0}\sum_{x''} \prod_{i=n_0}^\infty(1-p_{x''}) =0\,.
	\end{equation*}
	where the sum in $x'$ is finite. Thus, with probability one, in a finite number of steps we have strict inequalities reaching $x^*$, returning this as the solution. Thus, for all $x\in X$, every path starting at $x$ leads to $x^*$. Thus,
	\begin{equation*}
		\mathbb{E}_{\mu^\textnormal{A},\nu}(V\circ\Phi^\textnormal{A})\coloneqq \sum_{x\in X}\nu(x)\mathbb{E}_{\mu^\textnormal{A}}\left(V\circ\Phi^\textnormal{A}(x))\right)=1\,.
	\end{equation*}
	where $\Phi^\textnormal{A}(x)\coloneqq \phi^{\Phi^\textnormal{A}}.$
	
	Now, consider the diversity group $\Phi^\textnormal{D}$. This group is selected to maximize diversity and only allows minimal ability loss. However, there exists exactly one agent $\phi_{j_0}$ and one state $x^0$ such that $V\left(\phi_{j_0}(x^0)\right)<V\left(x^0\right)$. Let $x^{(-1)}\coloneqq \phi_{j_0}\left(x^0\right)$ and let $x$ such that $V(x)\le V(x')$. Similar as above, by fineness, the probability that
	\begin{equation*}
		\mathbb{P}\left(\exists~ i^x_k~\mid~\phi_{i^x_k}(x_{k-1})=x'\right)>0\,.
	\end{equation*}
	We have two possibilities:
	\begin{itemize}
		\item If $x_{k-1}=x^{(-1)}$ and $V(x)\le V(x^{(-1)})$, then, a ``disagreement cycle'' can be completed, $x_{k-1}=x^{(-1)}\to x^0\to x^{(-1)}$, returning $x^{(-1)}$. This happens with probability
		$$
		\sum_{k=1}^\infty \mathbb{P}\left(x_{k-1}=x^{(-1)}\right)\mathbb{P}\left(x_{k}=x^{0}\mid x_{k-1}=x^{(-1)}\right)\mu^\textnormal{D}_{x^0}\left(j_0\right)>0,
		$$
		where $\mathbb{P}\left(x_{k}=x^{0}\mid x_{k-1}=x^{(-1)}\right)=\mu^\textnormal{D}_{x^{(-1)}}\left(\{j\in J ~\mid~\phi_{j}\left(x^{(-1)}\right)=x^0\}\right)>0$, where we have used the full-diversity assumption.
		\item Also, if $x_{k-1}=x^{0}$ again, completes a disagreement cycle, $x^0\to x^{(-1)}\to x^0$, returning $x_0$. Similarly, this happens with probability
		$$
		\sum_{k=1}^\infty \mathbb{P}\left(x_{k-1}=x^{0}\right)\mu^\textnormal{D}_{x^0}\left(j_0\right)\mathbb{P}\left(x_{k+1}=x^{(-1)}\mid x_{k}=x^{(-1)}\right)>0\,.
		$$
		As $V(x^{(-1)})<V(x^0)<1$, thus,
		\begin{equation*}
			\mathbb{E}_{\mu^\textnormal{D},\nu}(V\circ\Phi^\textnormal{D})\coloneqq \sum_{x\in X}\nu(x)\mathbb{E}_{\mu^\textnormal{D}}\left(V\circ\Phi^\textnormal{D}(x))\right)<1\,.
		\end{equation*}

	\end{itemize}
\end{proof}
\section{The Diversity Prediction Theorem and the Crowds Beat Averages Law}\label{sec:5}
\subsection{The results.}
Hong and Page also present another theorem that would be useful to analyze. First, some definitions. Given a set of individuals labeled as $i=1,\ldots, N$, we associate to each of them a signal or prediction of some magnitude, which has $\theta$ as true value. The squared error of an individual's signal equals the square of the difference between the signal and the true outcome:
\[
\text{SE}(s_i) = (s_i - \theta)^2\,.
\]
The average squared error is given by
\[
\text{MSE}(\underline{s}) = \frac{1}{n}\sum_{i=1}^{n}(s_i - \theta)^2\,,
\]
with $\underline{s} \coloneqq (s_1, s_2, \ldots, s_n)$. The collective prediction is
\[
c =c(\underline{s}) = \frac{1}{n}\sum_{i=1}^{n} s_i\,.
\]
Predictive diversity of the collective is defined as:
\[
\hat{\sigma}(\underline{s}) = \frac{1}{n}\sum_{i=1}^{n}(s_i - c)^2\,.
\]
This is simply a (biased) estimation of the variance. Two trivial theorems can be deduced. The first, a particular version of the Pythagoras Theorem:
\begin{theorem}[Diversity Prediction Theorem]\label{th:div pred}
	The squared error of the collective prediction equals the average squared error minus the predictive diversity:
	\[
	\text{SE}\left(c\left(\underline{s}\right)\right) = \text{MSE}(\underline{s}) - \hat{\sigma}(\underline{s})\,.
	\]
\end{theorem}
\begin{proof}
	This is quite standard, but let us give a proof using the (generalized) Pythagoras Theorem. In $\mathbb{R}^n$ we can define the standard Euclidean or $l^2$-norm. If $\underline{c}=(c,\ldots, c)$ and analogously for $\underline{\theta}$, then $\langle\underline{s}-\underline{c}~,\underline{\theta}-\underline{c}\rangle_{l^2}=0$ so the Pythagoras Theorem gives
	\begin{equation}\label{eq: pyth}
		\|\underline{s}-\underline{\theta}\|_{l^2}^2=\|\underline{\theta}-\underline{c}\|_{l^2}^2+\|\underline{s}-\underline{c}\|_{l^2}^2\,.
	\end{equation}
\end{proof}
\begin{corollary}[Crowd Beats Averages Law]
	The squared error of the collective's prediction is less than or equal to the averaged squared error of the individuals that make up the crowd.
	\[
	\text{SE}\left(c\left(\underline{s}\right)\right) \leq \text{MSE}(\underline{s})\,.
	\]
\end{corollary}
\subsection{The asymmetric role of ``ability'' and ``diversity''.}\label{sec:asymmetric}
Before we proceed, let's note two simple mathematical observations:
\begin{error}\label{obs:trivial obs}
	\textbf{$MSE$ and $\hat{\sigma}$ cannot be treated as independent} as both depend on $\underline{s}$. That is,  altering one will generally change the other (it is not fixed), with the effect on the prediction error being, in principle, undetermined.
\end{error}
\begin{error}\label{obs:huge error}
	Therefore, it would be a significant mathematical error to consider that for the prediction error, $SE$, to be small is enough to make ``diversity'', $\hat{\sigma}$ large.  
\end{error}
These observations are mathematically trivial. Also, they can be graphically demonstrated when we consider the case of $n=2$, which brings us back to the standard Pythagoras theorem, see Figure \ref{fig:pythag_comparison}.
\begin{figure}
	\centering
	\begin{subfigure}{0.49\linewidth}
		\centering
		\includegraphics[width=\linewidth]{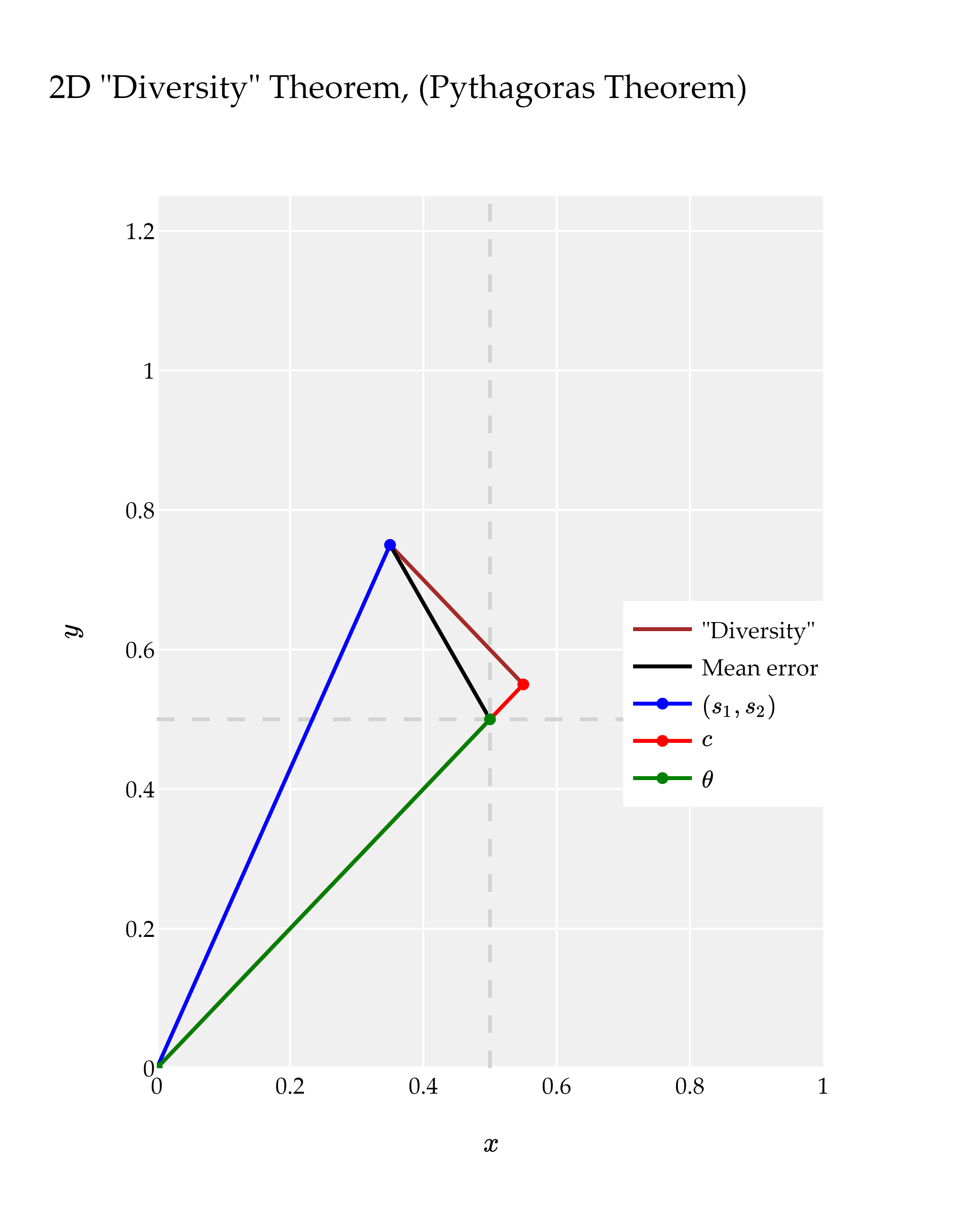}
		\caption{Initial result, prediction error is relatively ``small'': $\left(c-\theta\right)^2\approx0.0025$.}
		\label{fig:pythag1}
	\end{subfigure}
	\hfill
	\begin{subfigure}{0.49\linewidth}
		\centering
		\includegraphics[width=\linewidth]{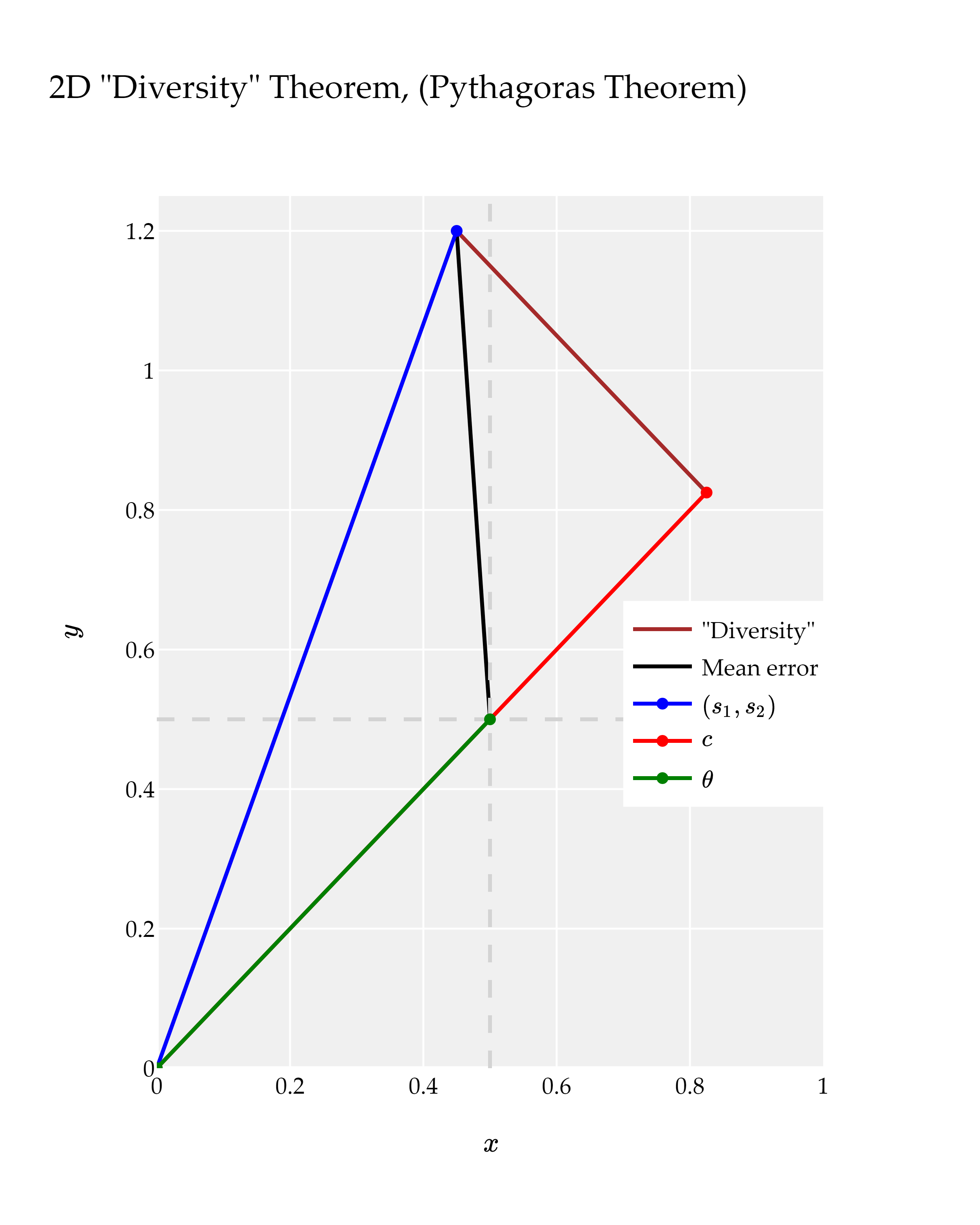}
		\caption{Result after increasing diversity, prediction error is relatively ``large'': $\left(c-\theta\right)^2\approx0.1056$.}
		\label{fig:pythag2}
	\end{subfigure}
	\caption{Increasing diversity does not always improve predictions and can sometimes significantly worsen them. The prediction error is represented by the red line, with the red dot indicating the prediction. The black line corresponds to $\text{MSE}$, and the brown line to $\hat{\sigma}$. The true value, $\theta$, is $\frac12$ (represented by the green dot). As diversity increases more than threefold (from 0.04 to 0.14), the squared error becomes more than forty times larger.}
	\label{fig:pythag_comparison}
\end{figure}
Knowing either $\text{MSE}(\underline{s})$ or $\hat{\sigma}$ alone is not sufficient to determine the value of the prediction error. In fact, according to the Crowd Beats Averages Law, we can see:

\begin{equation}
	\text{SE}\left(\text{MSE}, \hat{\sigma}\right) \in [0,\text{SE}^\textnormal{max}(\underline{s})]\,.
\end{equation}

This bound is sharp, with $\text{SE}^\textnormal{max}\coloneqq \text{MSE}$. Since $\text{SE}$ is not solely determined by either ``ability" or ``diversity", these variables can be observed in the context of the maximum prediction error, i.e., $\text{SE}^\textnormal{max}$. More precisely:

\begin{proposition}\label{prop:SEmax} Let $\text{SE}^\textnormal{max}$ represent the maximum prediction error. Then,
	\begin{itemize}
		\item If $\Delta\text{MSE}<0$, then $\Delta\text{SE}^\textnormal{max}<0$. In other words, if ``ability" increases, the maximum prediction error decreases. Particularly, if the increase in ability is large enough, the prediction error will decrease.
		\item If $\Delta\hat{\sigma}>0$, then $\Delta\text{SE}^\textnormal{max}\ge0$. This implies that if ``diversity" increases, the maximum prediction error also increases. In particular, an increase in diversity alone does not guarantee a reduction in the prediction error. Furthermore, if the increase in diversity is substantial enough, the maximum prediction error will also increase.
	\end{itemize}
\end{proposition}

\begin{proof}
	This is a trivial consequence of $\text{MSE}=\text{SE}^\textnormal{max}$ and the twin inequality of the Crowd Beats Averages Law: $\hat{\sigma}\le \text{SE}^\textnormal{max}$.
\end{proof}

Using the Crowd Beats Averages Law (and other trivial results), we arrive at a seemingly contradictory result: increasing ``ability" eventually reduces the prediction error, but increasing diversity ultimately increases the maximum prediction error. Consequently, the Diversity Prediction Theorem and the Crowd Beats Averages Law provide limited insight into how diversity impacts the prediction error in a general setting without controlling for ability.

\subsection{A basic mathematical error in advocating for diversity}\label{sec:Pages error}

Scott Page has argued that large diversity implies small prediction error. However, this conclusion, while favorable to the hypothesis that diversity reduces prediction error, constitutes a significant mathematical mistake. Indeed\footnote{Note that if one states that a large prediction error implies small diversity, by logical necessity, it is also stating that large diversity implies small prediction error (as a large prediction error would imply large diversity, a contradiction). }, in a \href{https://www.youtube.com/watch?v=EXn4vOuU3BE&list=WL&ab_channel=OsirisSalazar}{lecture (University of Michigan)}, Page states:

\begin{emphasisquote}
	And you might also ask, where does the madness of crowds originate? How could it be that a crowd could get something completely wrong? Well, that's not difficult to understand either, because crowd error equals average error minus diversity. If I want this to be large, if I want large collective error, then I need large average error, meaning that I need people to be getting things wrong, on average. \textbf{Additionally, I need diversity to be small}. So, \textbf{the madness of crowds comes from like-minded individuals} who are all incorrect, and once again, the equation provides us with this result.
\end{emphasisquote}

This mathematical misunderstanding involves a basic arithmetic error that we mentioned in Error \ref{obs:huge error}. From the ``Diversity Prediction Theorem" (with $\underline{s}$ term omitted for simplicity),
$$
\text{SE}= \text{MSE} - \hat{\sigma}\,,
$$
we \textbf{cannot} deduce that a large $\text{SE}$ implies a small $\hat{\sigma}$. Rather, it implies that $\text{MSE}$ must be much larger than $\hat{\sigma}$, where $\hat{\sigma}$ could be as large as desired. See, for instance, Figure \ref{fig:pythag2} for an illustration, where the prediction error is large, but diversity is larger (so it cannot be ``small").
\begin{remark}
	It is worth mentioning that Scott Page has made this point clear elsewhere. For instance, in his book ``The Diversity Bonus'', page 75, he says:
	\begin{quote}
		Our first intuition might be to select the three most accurate students. That group would have the lowest average error. Recall though that the group’s error depends in equal measure on their diversity. By choosing the group with the smallest average error, we ignore diver- sity. We ignore half of the equation. If the three most accurate people all think the same way, then we have no diversity. \textbf{Selecting the group of three people with the largest diversity makes even less sense}. To have high diversity, the group must have an even higher average error. The best approach is to search among all 570 possible groups for the one whose diversity is closest to its accuracy. 
	\end{quote}
\end{remark}

\section{Technical remarks}
To ease the exposition, several technical remarks are collected here.
\subsection{Remarks of Section \ref{sec:2}}
\begin{remark}\label{rem:Rawls}
	Assumption \ref{as:unique prob} appears to restrict the scope of problems to those of a technical nature, where the diversity of values might not be significant. However, the domains to which the theorem has been proposed for application, such as democracy, embody a profound component of political philosophy in which a diversity of values is unavoidable. Rawls, in particular, highlights "reasonable pluralism" as a fundamental element of his concept of a well-ordered society. Gaus further critiques this idea by arguing that reasonable pluralism not only encompasses varied conceptions of the good life but also extends to divergent views on justice itself. Gaus asserts \cite{Gaus10}:
	
	\begin{quote}
		Kant, Hare, and Rawls all believed that if we set up the constraints on moral legislation properly, the optimal eligible set will be a singleton. Only one rule can meet the constraints on right, properly understood. However, in our model, reasonable pluralism is taken seriously: Members of the Public are applying their interpretations of the test based on different evaluative standards and beliefs. There is no reason to suppose that, except at a very abstract level (§17), the results of the application of the same tests, by different Members of the Public, based on different evaluative standards, will yield the same answer. As we saw, having diverse evaluative standards, they will rank the proposals differently based on these standards. There is no reason to suppose that, out of all this diversity, Members of the Public will converge on the same proposal.
	\end{quote}
	
	In the context mentioned above, this could be interpreted as, according to Rawls, the $V_\phi$ functions may differ, yet there exists a unique optimum, $x^*$, such that $V_\phi(x^*)=1$ for all "reasonable" $\phi$, whereas for Gaus, not only do the $V_\phi$ functions do not equal a common $V$, but there also does not exist the singleton $\{x^*\}$. Consequently, in both cases,  the assumption is not met.
	
	This recognition of pluralism as an intrinsic aspect of social structure raises questions about the theorem's applicability without additional qualifications. Nonetheless, this analysis primarily examines the mathematical premises and their wider implications, temporarily setting aside discussions on value diversity to concentrate on the mathematical critique.
\end{remark}

\begin{remark}\label{rem:Phi0}
	In Theorem \ref{th:deterministic HP} two points warrant attention. Firstly, we could have consider a subset $\Phi_0$ such that $\forall ~x\in X$, agents reach $x^*$ consistent with Assumption \ref{as:group del}, i.e., $\mathbb{E}\left(V; \tilde{\Phi}_0\right)=1$, so  \eqref{eq:tilde Phi def} is trivially satisfied. For simplicity, one might prefer considering $\Phi$ over $\Phi_0$, as both groups consistently achieve the correct solution. Unneeded elements can be omitted without loss of generality. Given the uncertainty in the selection process, the only definitive method to include all indispensable agents is to encompass all. Whether we take into account $\Phi_0\subset \Phi_R$ or $\Phi\subset \Phi_R$ has no bearing on our arguments.
	
	Secondly, in \cite{HP04} and above, $N_1$ is an integer such that
	\begin{equation*}
		\mathbb{E}\left(V; \Phi_R\right)>\mathbb{E}\left(V; \{\phi^*\}\right)\,,
	\end{equation*}
	which translates to $\tilde{\Phi}_0$ being a subset of $\Phi$ capable of outperforming the best-performing agent. As we said, this subset's existence is a straightforward outcome of Corollary \ref{cor:diverse group}. Yet, selecting $N_1$ such that $\mathbb{E}\left(V; {\Phi}_0\right)=1$ provides a more encompassing result (as seen in Proposition \ref{prop:other results}), which inherently contains the result in \cite{HP04}; if an $N_1$ exists where the correct solution is invariably attained, then due to its monotonic nature, there must also be an $N_1$ where the alternative imperfect  group is outperformed. Again, the distinction between considering $\tilde{\Phi}_0\subset \Phi_R$ or $\Phi_0\subset \Phi_R$ is mostly irrelevant. For ease and based on the aforementioned discussion, we will also use $\Phi$ or $\Phi_0$ in subsequent sections. Further details can be found in Appendix \ref{ap:simpler proof}.
\end{remark}

\begin{remark}\label{rem:CJTvsHP} One could say that the Condorcet Jury Theorem (CJT) is as trivial as Theorem \ref{th:HP PNAS} and subject to similar critiques as the one presented above, but this would be an unfair comparison. The Hong-Page theorem and the CJT both try to contribute to our understanding of collective decision-making, yet they do so in markedly different ways. The Hong-Page theorem, despite its initial presentation of mathematical sophistication, as we have seen, essentially boils down to a self-evident proposition once its layers are peeled away. Its hypotheses—assuming agents never degrade a solution and that there exists always an agent who can improve upon it, coupled with the selection of agents until they surpass an already acknowledged imperfect best performer—do all the heavy lifting. Once these assumptions are made explicit, the theorem's conclusion becomes almost self-evident. Mathematics is not needed to understand that if it is assumed that the best agent is imperfect and the ``random group'' can be made perfect, the random group will outperform the best agent. Furthermore, when the theorem is adjusted to include more realistic characteristics, the original conclusion shifts dramatically, see Section \ref{sec:ability trumps div}, indicating the frailty of its insights under slight modifications. 
	
	In stark contrast, the CJT, even in its most basic form, delivers insights into the epistemic strength of democracy through collective decision-making. The theorem demonstrates how, under certain conditions, a large group of voters, each with a probability of being correct greater than chance, can collectively make decisions that converge to the correct choice with increasing certainty as the group size grows. While the proof of CJT utilizes probabilistic tools that may require university-level mathematics to fully grasp, it remains accessible and mathematics is essential to fully capture this insight. Moreover, the CJT has been the foundation for extensive refinement and expansion in political theory and economics literature, adapting to more complex scenarios such as correlated votes and non-homogeneous preferences \cite{Rom22} and references therein, using mathematics to understand the conditions under which the insight works, thereby enhancing its applicability and relevance\footnote{However, as we have criticized elsewhere \cite{Rom22}, the theorem is also often misused to reach conclusions that it cannot support.}.
	
	The comparison between the Hong-Page theorem and the CJT is thus not just in their mathematical complexity but in the depth of insight each provides. While the former wraps a simple fact in unnecessary mathematical additions and dramatically unstable under changes of the tailored hypotheses, the latter offers a compelling argument for the power of collective decision-making. Here, mathematics, in contrast to natural language, is essential for fully understanding the insight and the conditions for applicability in more complex cases. This contrast underscores the importance of critically examining the assumptions and implications of mathematical theorems in social sciences, distinguishing between those that offer genuine insights and those that can be reduced to simple facts when the unnecessary mathematics are removed.
\end{remark}
\subsection{Remarks of Section \ref{sec:3}}
\begin{remark}\label{rem:inj reply} (See Section \ref{sec:V inj})
	In real-life applications, the value of $V$ can be highly uncertain. Therefore, it is sensible to assume that, in the case of two states, $x,x'\in X$, where it is estimated that $V(x)\approx V(x')$, we set $V(x)=V(x')$ for practical purposes. This situation should not be disregarded as uncommon. Nevertheless, as argued in \cite{Kue17}, cited by Landemore:
	\begin{quote}
		You don’t fail to make it to the cashier in a grocery store when you are completely indifferent between buying one more apple or one more orange, nor do deliberators in a meeting fail to decide on some course of action if two options have precisely equivalent value. Adding a simple tie-breaking rule to the theorem is entirely sufficient to deal with the mathematical hiccup and move forward with the fundamental
		scientific question at hand.
	\end{quote}
	This argument does not address the technical point raised by Thompson. The problem is not that we are indifferent between the solutions $b$ or $c$, but rather that no one knows the solution if we do start at $b$ or $c$ (no one moves from these states to $d$). A tie-breaking rule only selects either $b$ or $c$, but never $d$, representing a failure of the theorem. The fact that the value function is ``indifferent'' implies that the hypotheses (in particular, the ``diversity'' assumption) are not sufficient to guarantee that $d$ is reached.
	
	The thesis of the theorem still holds if we replace Assumption \ref{as:diversity} with $\forall x\in X\backslash\{x^*\}~\exists~ \phi\in\Phi ~/~ V\left(\phi(x)\right)> V(x)\,$. However, this adjustment only serves to render the theorem more trivial and misapplies the term diversity. This condition simply implies that for every state, there exists an agent that can strictly improve that state. It is unsurprising that in a finite number of steps, these agents reach the maximum, which, by hypothesis, the best problem solver cannot always attain. Consequently, this adjustment does fix the theorem, but at the cost of making it more trivial and highlighting that what the theorem requires is not ``diversity", but the existence of a more ``able" problem solver who can improve upon areas where others fall short.
\end{remark}
\begin{remark}\label{rem:X inf}(See Section \ref{sec:clone performance})
	Regarding the issue of clones, the following is stated in \cite{HP04}:
	\begin{quote}
		[...] we present a simpler version of our result where X is
		assumed to be finite. This finite version makes the insight more straightforward, although it comes at the cost of trivializing some intricate assumptions and arguments. For example, the group of the best-performing agents is proven below to be comprised of identical agents. This is an artifact of the finite version. In the general version under reasonable conditions, the group of the best-performing agents can be shown to be similar, not necessarily the same.
	\end{quote}
	However, this explanation is far from accurate. Clones also appear in the less realistic case where $X$ is not finite. This occurs because we have to take copies from $\Phi$ and, if $\phi$ has already appeared, it can appear again\footnote{With positive probability in the reasonable case of a finite $\Phi$.}. Moreover, the finite version of the model is neither sufficient nor necessary for proving that the group of the best-performing agents is comprised of identical agents.
	
	In a scenario where $X$ is finite, the best agents could be several different ones once we consider only one copy of each agent. This could be easily demonstrated by following my previous example from Section \ref{sec:unique best} or \cite[Assumption 5']{HP98}. In the version where $X$ is not finite, according to Assumption 5 of their appendix, $B(\phi^*,\delta)\cap\Phi=\{\phi\in\Phi~|~d(\phi,\phi^*)<\delta\}$ could contain only one agent, namely $\phi^*$. It should also be noted that a finite $X$ represents a more realistic setup. Typically, rendering things continuous simplifies the analysis, as it allows us to use standard calculus, for example, but this is not the case here. It is less realistic to assume that agents have answers to an infinite set of elements than to a finite set. 
	
	Furthermore, to reach the conclusion in this case, more predetermined hypotheses are needed, which cannot be justified in terms of being reasonable or intuitive. Basically, Assumption 5 in \cite[Appendix]{HP98} imposes that the best agents are clustered near the best performing agent. They are clustered, in the sense of the distance metric defined above, in a way that guarantees that if they are close, the expected performance will be close too, \cite[Lemma 2]{HP98}. Then, in the proof, they choose the best agents as close to the best performing agent as needed so that they do not always reach \(x^*\), \cite[Lemma 4]{HP98}.
\end{remark}

\newpage
\bibliographystyle{alpha}
\bibliography{HPLBib.bibtex}
	
 \end{document}